\documentclass[journal,final,twocolumn,10pt]{IEEEtranTCOM}
\normalsize

\ifCLASSINFOpdf
\else
\fi

\usepackage{graphicx}

\usepackage{color}
\usepackage{amsfonts}
\usepackage{amsmath}
\usepackage{amsthm}
\usepackage{amssymb}
\usepackage{alltt}
\usepackage{mathrsfs}
\usepackage{subfigure}
\usepackage{esint}
\usepackage{enumerate}
\usepackage[noadjust]{cite}
\usepackage{float}

\usepackage{hyperref}
\hypersetup{hidelinks}

\DeclareGraphicsExtensions{.eps}
\usepackage{soul}
\usepackage{array}
\usepackage{todonotes}
\usepackage{bm}

\usepackage{algorithm}
\usepackage{algpseudocode}

\newtheorem{proposition}{Proposition}
\newtheorem{lemma}{Lemma}

\newcommand{\mb}[1]{\mathbf{#1}}
\newcommand{\mc}[1]{\mathcal{#1}}
\newcommand{\mbb}[1]{\mathbb{#1}}

\newlength{\figWidth}
\newcounter{MYtempeqncnt}

\begin{document}
%
\title{Secrecy Rate of Cooperative MIMO in the Presence of a Location Constrained Eavesdropper}
%
%
%

\author{Zhong~Zheng,~\IEEEmembership{Member,~IEEE,}
        ~Zygmunt~J.~Haas,~\IEEEmembership{Fellow,~IEEE,}
        ~and~Mario~Kieburg
\thanks{Z. J. Haas has been funded by ARL Grant number W911NF1810406. Z. Zheng and Z. J. Haas have been funded in part by the NSF grant number ECSS-1533282. M. Kieburg acknowledges support by the German research council (DFG) via the CRC 1283: \lq\lq Taming uncertainty and profiting from randomness and low regularity in analysis, stochastics and their applications\rq\rq.


Z. Zheng was with the Department of Computer Science, University of Texas at Dallas, Richardson, TX 75080 USA (e-mail: zhong.zheng.z@ieee.org). Z. J. Haas is with the Department of Computer Science, University of Texas at Dallas, Richardson, TX 75080 USA, and also with the School of Electrical and Computer Engineering, Cornell University, Ithaca, NY 14853 USA (e-mail: zhaas@cornell.edu). M. Kieburg is with the Faculty of Physics, Bielefeld University, PO-Box 100131, 33501 Bielefeld, Germany (e-mail: mkieburg@physik.uni-bielefeld.de).}}

%
%

\markboth{IEEE Transactions on Communications}%
{Submitted paper}
%



\maketitle
\begin{abstract}
We propose and study the cooperative MIMO architecture to enable and to improve the secrecy transmissions between clusters of mobile devices in the presence of an eavesdropper with certain location constraint. The cooperative MIMO system in this paper (referred to as \emph{Reconfigurable Distributed MIMO}) is formed by temporarily activating clusters of nearby trusted mobile devices, with each cluster being centrally coordinated by its corresponding cluster head. We assume that the transmitters apply a practical eigen-direction precoding scheme to transmit the confidential signal and artificial noise, while the eavesdropper can be located in multiple possible locations in the proximity. We first obtain the expression of the secrecy rate, where the required average mutual information between the transmitters and the eavesdropper is characterized by closed-form approximations. The proposed approximations are especially useful in the secrecy rate maximization, where the original non-convex problem can be solved by the successive convex approximations. Numerical results show that the secrecy rate can be significantly improved by leveraging the location constraint of the eavesdropper, compared to the existing result. We also demonstrate that the secrecy rate can be further improved by increasing the cluster size.
\end{abstract}



%
\IEEEpeerreviewmaketitle

\section{Introduction}
%
%
%
%

\IEEEPARstart{P}{hysical-layer} security has recently attracted increasing attention in the field of wireless communications, as it guarantees information security between the legitimate entities even when the eavesdropper overhears the communications~\cite{liang2009}. Unlike the conventional cryptographic methods, the physical-layer security does not rely on key-based encryption for the confidential message. Instead, the legitimate transmitter utilizes the characteristics of the communication channels to encode the information, so that the confidential message is conveyed to the legitimate receiver, while the information leakage to the eavesdropper is eliminated in the information-theoretical sense~\cite{wyner1975}. Therefore, the physical-layer security guarantees perfect information secrecy between legitimate transceivers and circumvents the challenges in the conventional key-based encryption, such as key distribution and management, which are especially critical in wireless communication due to its broadcasting nature. 

Throughout this paper, we refer to the legitimate transmitter(s) and the legitimate receiver(s) as Alice and Bob, respectively, and refer to the eavesdropper as Eve. Under the secrecy constraint, the maximum information rate between Alice and Bob is characterized by the \lq\lq secrecy capacity.\rq\rq\ As shown in~\cite{wyner1975} and~\cite{csiszar1978}, the secrecy capacity is the difference between the Shannon rates of two channels, the legitimate channel between Alice and Bob and the eavesdropping channel between Alice and Eve, under the condition that the mutual information between Alice and Eve is zero. 

Multi-antenna transceivers increase the spatial degrees of freedom of the communication, which can be leveraged to improve the secrecy rate~\cite{hero2003}. When Alice has a single antenna and Bob has multiple antennas, a.k.a. the Single-Input Multi-Output (SIMO) system, the secrecy capacity and the outage secrecy capacity were investigated in~\cite{parada2005} and~\cite{zhu2016}, while the corresponding MISO case was studied in~\cite{li2007} and~\cite{shafiee2007}. The secrecy capacity of a special case of the MIMO wiretap channel was considered in~\cite{shafiee2009}, where Alice and Bob are equipped with dual antennas, and Eve is equipped with a single antenna. Recently, the physical-layer security has also been considered for the Distributed-MIMO (D-MIMO) systems in~\cite{jin2013,zhang2013,ng2015,akino2011,wang2016}, where Alice has a distributed antenna array composed of remote fixed antenna ports or base stations. Therein, the distributed antenna elements are connected to the central base station via tethered high-speed fiber links, which is relevant for infrastructure D-MIMO systems. The D-MIMO configurations further improve the secrecy capacity by leveraging the macro diversity due to more diverse path losses and shadowing~\cite{wang2016}. Unlike the infrastructure D-MIMO systems, the cooperative MIMO architecture studied in this paper is based on forming the distributed antenna arrays on-demand, where each array is composed of a cluster of nearby legitimate mobile devices. The cooperative MIMO improves the secrecy rate by leveraging the benefits of distributed MIMO, while having the flexibility in the topology and the size of the formed antenna arrays.

\subsection{Related Works}

Evaluating the secrecy capacity amounts to finding the optimal transmitted signals of Alice, which maximizes the difference between the mutual information of the legitimate and the eavesdropping channels. In the MIMO wiretap channels, the problem becomes optimizing the input covariance matrix, which is a difficult non-convex problem. Previous works showed that the design of the optimal input heavily depends on the antenna configurations of the MIMO channels and the amount of Channel State Information at Transmitter (CSIT) available at Alice. In particular, when both Alice and Eve are equipped with multiple antennas, Bob has a single antenna, and Alice has CSITs of both channels, authors in~\cite{khisti2010} show that the capacity-achieving transmitter scheme is beamforming. The beamforming direction is chosen as the generalized eigenvector corresponding to the maximum generalized eigenvalue of the wiretap channel. When Bob is also equipped with multiple antennas, the conditions of a full rank and a rank deficient input covariance matrix are proved in~\cite{oggier2011}. In~\cite{khisti2010b} and~\cite{bustin2009}, assuming general multi-antenna wiretap channels, the input covariance matrices are constructed via the generalized eigenvalue decompositions of the main channel and the eavesdropping channel. The closed-form expression of the optimal input is given in~\cite{fakoorian2013} when the input covariance matrix is full rank, and is given in~\cite{loyka2016} when it is either full rank or of rank-one.

In the context of secrecy communications, the eavesdropper is usually passive and silent as to avoid being detected. Therefore, it is more practical to assume only the statistical, rather than the full knowledge, of the CSIT of the eavesdropping channel~\cite{wang2015}. Compared to the full CSIT counterpart, there are much fewer results available for the secrecy capacity under the statistical CSIT assumption. The optimal transmitter can be only characterized for certain channel configurations. When Alice has multiple antennas, and both Bob and Eve have single receive antenna,~\cite{li2011} provides a sufficient condition for the optimal input covariance matrix being rank-one. In the same setting, authors in~\cite{lin2013} derive the optimal input covariance matrix, where the artificial noise is injected and optimized. The optimal rank-one input covariance matrix has been also identified in~\cite{li2011b} when Eve is equipped with multiple receive antennas, and an on-off power allocation scheme is proposed in~\cite{nguyen2011} to maximize the secrecy rate. When the statistical CSITs of both channels are available and Bob has more receive antennas compared to Eve, authors in~\cite{lin2014} show that the optimal input covariance matrix is an identity matrix with uniform power allocation across the transmit antennas. In~\cite{wang2015b}, a random precoding is applied at Alice and is compared with artificial noise injection scheme. However, when the instantaneous CSIT of the legitimate channel and the statistical CSIT of the eavesdropping channel are available, the capacity-achieving transmitter design is still unknown for generic antenna configurations. 

Alternatively, the secrecy rate maximization, assuming statistical CSIT of the eavesdropping channel, has been pursued by approximating the information rate between Alice and Eve, which leads to a simplified optimization problem. In~\cite{vishwakarma2014}, the average mutual information between Alice and Eve is upper-bounded by a deterministic channel and the optimal precoding has the same direction as the generalized eigenvector of the approximate wiretap channels. In~\cite{wu2012}, the average mutual information is lower-bounded with a simplified analytical expression, and the upper bound of secrecy rate is obtained by assuming linear precoding at the transmitter. By approximating the average rate of Eve with its Taylor series expansion, the non-convex secrecy rate maximization reduces to a sequence of convex sub-problems~\cite{zappone2016}. However, the techniques used in~\cite{vishwakarma2014,wu2012,zappone2016} are relevant to the MIMO channels with co-located antenna arrays, and cannot be applied to the D-MIMO (or RD-MIMO) channels, as we assumed in this paper. 

Care must be taken in the secrecy rate optimization for D-MIMO system as the topology of the spatially distributed antennas needs to be incorporated into the optimization problem. Due to this additional technical complication, the existing works on the secrecy D-MIMO systems either assume full CSIT of the eavesdropping channels~\cite{jin2013,zhang2013,ng2015} or adopt simplified transmission methods with statistical CSIT. In~\cite{akino2011}, assuming the linear Time-Reversal (TR) transmissions from distributed antenna systems, the asymptotic secrecy rate is obtained in the large and low SNR regimes, respectively. In~\cite{wang2016}, assuming multiple distributed transmitters with a sub-optimal diagonal precoding, the secrecy rate maximization is asymptotically approximated as a max-min problem for the number of antennas approaching infinity. 

In~\cite{jin2013,zhang2013,ng2015,wang2016,khisti2010,oggier2011,khisti2010b,bustin2009,fakoorian2013,loyka2016,li2011,lin2013,li2011b,nguyen2011,lin2014,wang2015b,vishwakarma2014,wu2012,zappone2016}, the investigation of the secrecy wireless communications assumes a single eavesdropper's location known to Alice. Alice optimizes the secrecy transmissions according to this prior knowledge, and therefore, the resulting secrecy rate is only valid for the presumed eavesdropper's location. To relax this constraint, authors in~\cite{zhu2013,goel2008,zhou2010,he2014} assumed a noiseless eavesdropper with artificial noise injected by the transmitter. This assumption allows the eavesdropper to be anywhere around Alice and the corresponding secrecy rate is valid regardless of the exact location of Eve. However, in some communication scenarios, the \lq\lq anywhere Eve\rq\rq\ is an overly conservative assumption that leads to pessimistic estimation of the secrecy rate. Indeed, in many practical situations, Eve cannot be located in certain areas, thus allowing to implement a larger secrecy rate.

A relevant physical-layer security technique is based on the cooperative communication networks, where a group of trusted nodes are either used as relays or jammers, see e.g.,~\cite{wang2015} and~\cite{wang2015c}. In particular, a cooperative relaying scheme is proposed in~\cite{xu2017} and~\cite{wang2018} to improve the received SNR of the legitimate receiver, where the relays are selected according to the degree of social tie with the legitimate transmitter. Therein, the SNR thresholds serve as the secrecy and reliability criterions at the eavesdropper and at the receiver, respectively, and the locations of the eavesdroppers are modeled as statistical point processes. 

\subsection{Physical-Layer Security of Cooperative MIMO}
In our recent work~\cite{zheng2017}, we studied a mobile cooperative MIMO architecture, which we refer to as the \emph{Reconfigurable Distributed MIMO (RD-MIMO)}. Therein, a group of nearby mobile nodes forms a node cluster to jointly transmit or receive wireless signals. The cluster is coordinated by a head node and the coordination signaling is exchanged within the cluster via local wireless connections. When two clusters communicate with each other, and the head node of each cluster jointly encodes or decodes the communicated symbols, the distributed antenna arrays resemble the D-MIMO system. In this work, we consider the physical-layer security of the communications between such two clusters of legitimate nodes, each collectively called Alice and Bob, respectively. The inter-cluster transmissions are wiretapped by an eavesdropper, which may be placed at a number of possible locations, including those near the legitimate devices. The secrecy rate optimization with regard to the multiple possible locations generalizes the conventional framework that assumes a single eavesdropper's location. This is also useful to relax the \lq\lq anywhere Eve\rq\rq\ assumption in~\cite{zhu2013,goel2008,zhou2010,he2014}, and the considered framework can take into account practical constraints of the system topology to improve the achievable secrecy rate. Namely, if there exists an area around Alice guaranteed to be free of eavesdropper, the secrecy rate can be optimized assuming that Eve is indeed placed outside this eavesdropper-free region. For instance, it could be assumed that Eve could not be located within a secure military camp, allowing the devices within the camp to achieve significantly larger secrecy rate.

The RD-MIMO framework improves the secrecy rate both effectively and flexibly. On one hand, compared to the secured transmissions between two single devices, the spatial degrees of freedom of RD-MIMO are increased, which can be utilized in designing the secrecy transmission. On the other hand, the node cluster can be formed on-demand and the number of nodes in the cluster can be determined according to the performance requirement.

We summarize the contributions of the proposed secrecy cooperative MIMO framework:
\begin{itemize}
\item We assume Eve may appear at a finite number of possible locations. As the eavesdropper may keep silent, Alice only has the statistical CSITs of the eavesdropper's channels corresponding to each of the possible Eve's locations. This system model can be used to evaluate the achievable secrecy rate of the cooperative MIMO, where the location of Eve is loosely constrained without specifying a known location\footnote{In order to fulfill certain location constraint, Eve should be located within a continuous region or on a continuous contour. However, as will be shown in Section~\ref{secOpt}, the region or the contour constraint can be well-approximated by placing Eve on a sufficiently large number of discrete possible locations. We will leave the study of a region-constrained Eve in future work.}. For example, Eve may be separated from Alice with a certain minimum distance.
\item We consider eigen-direction precoding (\cite{shafiee2007}) to construct the input covariance matrix, which reduces the original problem to a lower dimensional power allocation problem. Accordingly, we derive approximations for the average rate between Alice and Eve as a function of the power allocation vector. Numerical results show that the proposed approximations are reasonably accurate and that they provide over-estimates for the exact average rate with a large probability. This is relevant in the context of physical-layer security to fulfill the required secrecy constraint. In addition, compared to the exact expression of the average rate between Alice and Eve, the proposed approximations are explicit functions of the power allocation vector. This is a useful property to further approximate the average rate, a concave function of the power allocation vector, with a linear affine function, which greatly simplifies the secrecy rate maximization.
\item The non-convex secrecy rate maximization, assuming a number of possible locations of Eve, is solved by successive convex approximations~\cite{marks1978}. It converts the original problem into a sequence of convex sub-problems using the affine approximations for the average rate between Alice and Eve. Each sub-problem can be efficiently solved by standard convex optimization tools. Numerical results show that the proposed optimization framework, by incorporating the location constraint of Eve, can eliminate the outage of transmissions and further improve the achievable secrecy rate, as compared to the existing result of the conservative \lq\lq anywhere Eve\rq\rq\ assumption. In addition, the secrecy rate can be also improved by activating cooperative nodes having a more diverse channel conditions due to distributed node placements. 
\end{itemize}

The rest of the paper is organized as follows: Section~\ref{secModel} introduces the signal model of cooperative MIMO wiretap channels and outlines the eigen-direction precoding scheme. In Section~\ref{secAvgRate}, we derive approximations for the average rate between Alice and Eve using the eigen-direction precoding. Numerical results are illustrated to validate the accuracy of the proposed approximations. In Section~\ref{secOpt}, we present the secrecy rate maximization framework and the numerical results. In Section~\ref{secConclude}, we conclude the findings of this paper. 

\section{System Model}\label{secModel}

Consider a pair of legitimate transmit and receive nodes, which are hereafter referred to as the transmit and receive head nodes. Confidential information is transmitted between the head nodes over the wireless channel and prone to be eavesdropped on by a malicious listener, as shown in Fig.~\ref{figSystem}. In the circular area centering at the transmit head node with radius $r_T$, we assume there is a cluster with $K-1$ trusted transmit nodes, and the transmit cluster is collectively referred to as Alice. Similarly, in the circular area centering at the receive head node with radius $r_R$, there are $N-1$ trusted receive nodes, and the receive cluster is collectively referred to as Bob. All the legitimate and trusted nodes are equipped with a single antenna, while the malicious listener, Eve, is equipped with an $M$-antenna array. Note that there are various ways to identify trusted nodes near the head nodes. As an example, the trusted nodes can be identified by the degree of social tie with the cluster head nodes~\cite{xu2017,wang2018}. In this work, we do not elaborate on the node identification process, but we assume the trusted nodes are given. 

\begin{figure}[t!]
\centerline{\includegraphics[width=2.8in,height=1.9in]{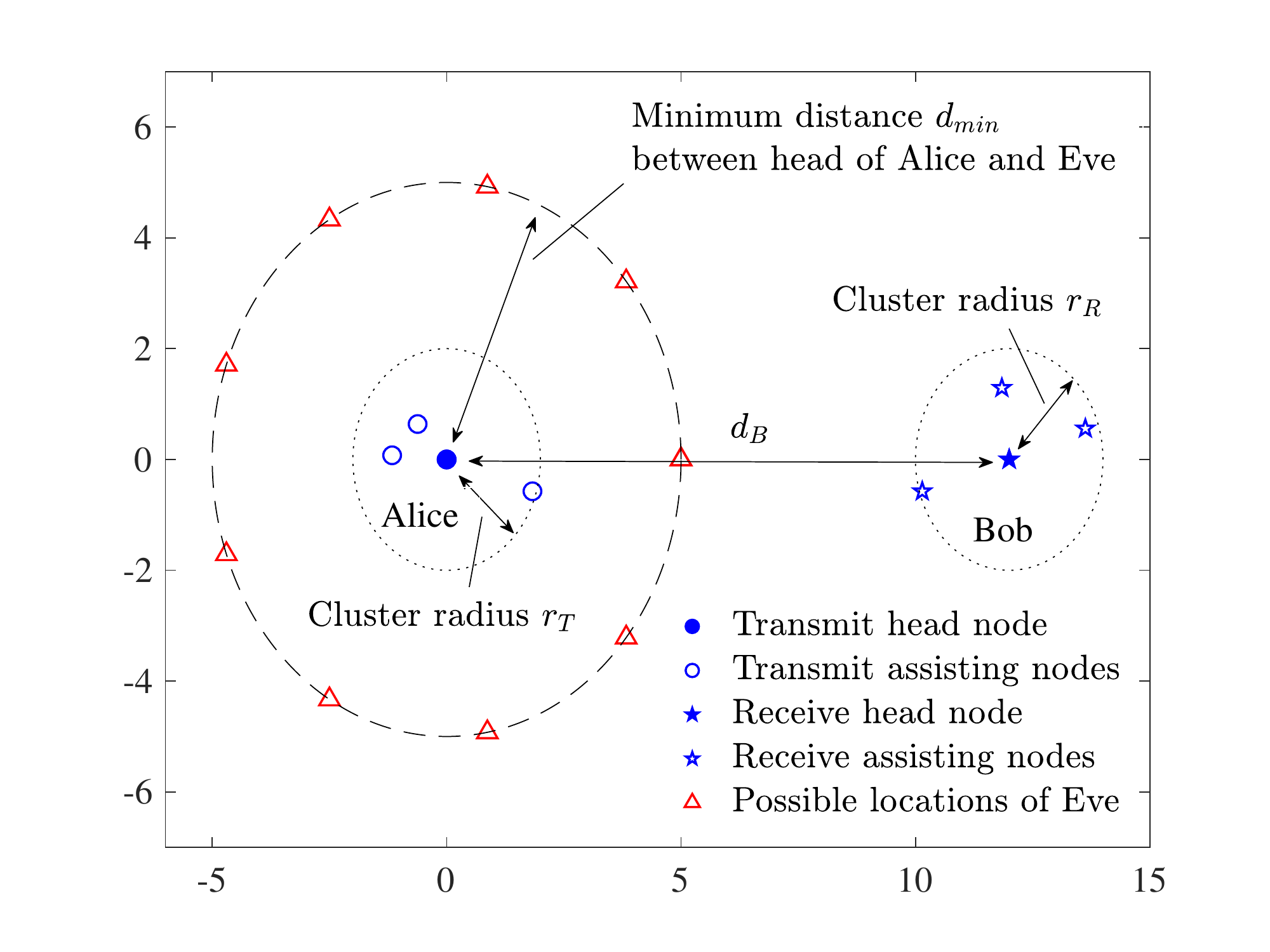}
}
\caption{Cooperative MIMO in presence of an eavesdropper. Possible locations of Eve have a minimum distance $d_{\mathrm{min}}$ towards the head of Alice. The distance between head nodes of Alice and Bob is $d_B$.}
\label{figSystem}
\end{figure}

In typical communication systems, the eavesdropper remains silent to avoid being detected by the legitimate nodes. However, in certain scenarios, it is relevant to assume that the position of the eavesdropper is constrained to possible locations, and that the communication between the legitimate nodes can be optimized according to the prior knowledge of such possible locations. As an example, consider a communication system with a minimum distance $d_{\mathrm{min}}$ guaranteed between the transmit head node and the eavesdropper. To obtain a non-trivial lower bound on the achievable secrecy rate, the transmitting scheme can be optimized by assuming that the eavesdropper is located on a circle centered at the transmit head node with radius $d_{\mathrm{min}}$. When the transmitter implements a code with data rate less or equal to the lower bound of secrecy rate, the communications satisfy the secrecy criterion, when the eavesdropper is anywhere outside the $d_{\mathrm{min}}$-radius circle\footnote{The eavesdropping model can be viewed as a special case of the annulus threat model as in~\cite{yan2015}.}. Fig.~\ref{figSystem} shows the possible locations of the eavesdropper, evenly distributed on a circle. The use of sampled discrete locations reduces the complexity of the secrecy rate optimization. As shown in Section~\ref{secOpt}, the approximate but highly likely lower bound of the secrecy rate can be found by assuming 10 evenly distributed locations on the circle. The minimum distance restriction is practical in many communication scenarios. In addition to the military camp example, other examples include cooperative transmissions from within a residential house or from within a personal vehicle, where multiple trusted communication devices belong to a private person or persons with high degree of social ties~\cite{wang2018}. Without permission to physically enter the camp, the house, or the vehicle, any eavesdropper has to be located outside the perimeter of the premises, and therefore, a certain minimum distance is guaranteed. Note that the following discussions and the optimization framework in Section~\ref{secOpt} can be also extended when the possible locations of Eve are placed in an arbitrary area.

Using the RD-MIMO framework as discussed in~\cite{zheng2017}, the $K$ transmit nodes form a cooperative transmit cluster, where the head node is responsible for encoding the information symbols into transmit signals, distributes the encoded signals to the cluster nodes, and synchronizes the transmissions from the cluster. Similarly, the $N$ receive nodes form the cooperative receive cluster, in which the receive head node collects the received signals from its assisting nodes and jointly decodes the receive symbols. We assume the node cooperation within each cluster is performed over high-speed wireless connections, while the transmissions between clusters have much lower rate due to more severe channel impairments. Therefore, the inter-cluster transmissions are the bottleneck of the system, and the channels between Alice and Bob resemble the D-MIMO channels. We focus on the physical layer security of the inter-cluster communication, where the communication between Alice and Bob in the presence of Eve is modeled by the MIMO wiretap channel~\cite{loyka2016}. We assume the security of the intra-cluster communication within Alice and Bob can be guaranteed relatively easily, because the cluster nodes are in close proximity, so the communication links have shorter range and have higher channel capacity. 

\subsection{Signal Model}\label{secSignal}

Given a vector $\mb{x}=[x_1,\ldots,x_K]^\mathrm{T}$, where $x_k$ denotes the transmit signal of the $k^{th}$ node at Alice, the vector $\mb{y}=[y_1,\ldots,y_N]^\mathrm{T}$ denotes the receive signals at Bob, and $\mb{z}_i=[z_{i,1},\ldots,z_{i,M}]^\mathrm{T}$, $1\le i\le L$, denotes the receive signals at Eve when Eve resides at the $i^{th}$ location: 
\begin{align}
\mb{y} &= \sqrt{g_B}\,\mb{H} \mb{x} + \mb{n}_B,\\
\mb{z}_i &= \sqrt{g_E}\,\mb{F}_i \mb{x} + \mb{n}_E, 
\end{align}
where $L$ is the number of possible locations of Eve, $y_n$ and $z_{i,m}$ are the receive signal of the $n^{th}$ node at Bob ($1\le n\le N$) and of the $m^{th}$ antenna at Eve ($1\le m\le M$), respectively. The MIMO channel between Alice and Bob is denoted as~$\mb{H}$, where the entry $\mb{H}_{n,k}$ is the channel coefficient between the $k^{th}$ transmit node in Alice and the $n^{th}$ receive node in Bob. The channel coefficient $\mb{H}_{n,k}$ is modeled as the complex Gaussian random variable with zero mean and variance $\beta_{n,k}$, where $\beta_{n,k}$ is the normalized average channel gain given by:
\begin{equation}
\beta_{n,k} = \frac{\mathrm{PL}(d_{n,k})}{g_B} = \left(\frac{d_B}{d_{n,k}}\right)^\alpha.
\end{equation}
Here, $\alpha(\ge 2)$ is the path-loss exponent, $\mathrm{PL}(d)=c_p d^{-\alpha}$ is the average channel gain between two nodes at a distance $d$, $c_p$ is the path loss at the unit distance, and $g_B=\mathrm{PL}(d_B)$ denotes the average gain between the head nodes of Alice and Bob.

The channel coefficients between Alice and Eve are denoted as $\mb{F}_i=\mb{W}_i \mb{\Sigma}_i^{1/2}$, where the $M\times K$ matrix $\mb{W}_i$ denotes the fast fading coefficients, modeled by \emph{i.i.d.} standard complex Gaussian random variables. The $K\times K$ diagonal matrix $\mb{\Sigma}_i=\mathrm{diag}([\sigma_{i,1},\ldots,\sigma_{i,K}])$ denotes the average channel gains between Alice and Eve with the $k^{th}$ diagonal entry $\sigma_{i,k}$ being: 
\begin{equation}
\sigma_{i,k} = \frac{\mathrm{PL}(d_{i,k})}{g_E} = \left(\frac{d_E}{d_{i,k}}\right)^\alpha,\label{eqsigma}
\end{equation}
where $d_{i,k}$ is the distance between the $k^{th}$ node of Alice and the $i^{th}$ location of Eve, and $g_E=\mathrm{PL}(d_E)$ denotes the average channel gain between the head node of Alice and Eve. The additive noise $\mb{n}_B$ at Bob and $\mb{n}_E$ at Eve are modeled as \emph{i.i.d.} complex Gaussian vectors with power $N_0$, i.e., $\mb{n}_B\sim \mc{CN}(\mb{0},N_0 \mb{I}_N)$ and $\mb{n}_E\sim \mc{CN}(\mb{0},N_0 \mb{I}_M)$.

In this work, we adopt the following signal-level assumptions:
\begin{enumerate}
\item The channels $\mb{H}$ and $\mb{F}_i$ are frequency flat and follow the block fading process, i.e., the entries of $\mb{H}$ and $\mb{F}_i$ vary independently from one channel coherent time to another, but remain constant for each coherent time.
\item The instantaneous CSI of $\mb{H}$ is known by the nodes in Alice, Bob, and Eve. The instantaneous CSI $\mb{F}_i$ is only known by Eve, while the statistical CSIs of $\mb{F}_i$, $i=1,\ldots,L$, are available to Alice.
\end{enumerate}

The statistical CSI of $\mb{F}_i$ relies on the distance-dependent path loss and the number of available antennas at Eve, which is acquired from prior knowledge. For example, there may exist a maximum number of antenna $M_\mathrm{max}$ that Eve can be equipped with and a minimum distance $d_{\mathrm{min}}$ between the head node of Alice and Eve. By setting $M=M_\mathrm{max}$ and $d_E=d_\mathrm{min}$, a lower bound of the secrecy rate can be determined. The achievable secrecy rate and a practical precoding scheme are discussed in the next subsection.

\subsection{Secrecy Rate and Precoder Design}

The considered wiretap channel model with multiple possible locations of Eve is equivalent to the compound wiretap channel with the receiver CSI~\cite{bloch2013}. Denote the channel input as $\mc{X}$, the channel output at the legitimate receiver as $\mc{Y}$, and the output at the $i^{th}$ eavesdropper as $\mc{Z}_i$, $1\le i\le L$. According to~\cite[Prop. 6]{bloch2013}, the following secrecy rate $\widehat{R}_{\mathrm{sec}}$ is achievable:
\begin{equation}
\widehat{R}_{\mathrm{sec}} = \max_{p(\mc{V,\mc{X}})} [I(\mc{V};\mc{Y}) - \max_{1\le i\le L} I(\mc{V};\mc{Z}_i)]^+,\label{eqRh_sec}
\end{equation}
where $[x]^+ = \max(0,x)$, $I(a;b)$ denotes the mutual information between the random variables $a$ and $b$, $\mc{V}$ is an auxiliary random variable, and $\mc{V}\rightarrow \mc{X}\rightarrow (\mc{Y},\mc{Z}_i)$, $1\le i\le L$, form Markov chains. The outer maximization in (\ref{eqRh_sec}) is optimized over the joint distribution $p(\mc{V},\mc{X})$ of the random variables $\mc{V}$ and $\mc{X}$. Note that the secrecy rate (\ref{eqRh_sec}) is derived in~\cite{bloch2013} assuming various secrecy criterions, where each corresponding secrecy outage probability decays exponentially to zero as the length of codeword increases to infinity. Since the CSIs are available to the receivers, the outputs of the compound channel at the eavesdropper can be viewed as the tuples $\mc{Z}_i=\{\mb{z}_i,\mb{F}_i\}$. By using the chain rule of mutual information, we obtain: 
\begin{align}
I(\mc{V};\mc{Z}_i) &= I(\mc{V};\mb{z}_i,\mb{F}_i) \nonumber\\
&= I(\mc{V};\mb{z}_i|\mb{F}_i) + I(\mc{V};\mb{F}_i) = I(\mc{V};\mb{z}_i|\mb{F}_i),
\end{align}
where $I(\mc{V};\mb{F}_i)=0$ since the instantaneous channel is unknown to the transmitter and therefore $\mc{V}$ and $\mb{F}_i$ are independent. Similarly, $I(\mc{V};\mc{Y})=I(\mc{V};\mb{y},\mb{H})$ by treating $\mc{Y}=\{\mb{y},\mb{H}\}$.	

The rate maximizing distribution $p(\mc{V},\mc{X})$ in (\ref{eqRh_sec}) is still an open problem for generic MIMO wiretap channel. To proceed, we set $\mc{V}=\mb{x}$ and assume Gaussian signaling at the transmitter, i.e., $\mb{x}$ is a multivariate Gaussian vector. In order to obscure the information reception at Eve, artificial noise is sent with the information signal, and the transmitted symbol $\mb{x}$ is the sum of the precoded information and the artificial noise, i.e.,
\begin{equation}\label{eqx}
\mb{x} = \sqrt{P}\,\mb{V}_s\mb{\Psi}_s^{1/2}\mb{s} + \sqrt{P}\,\mb{V}_a\mb{\Psi}_a^{1/2}\mb{a},
\end{equation}
where $P$ denotes the total maximum transmit power, $\mb{s}\sim\mc{CN}(\mb{0},\mb{I}_K)$ and $\mb{a}\sim\mc{CN}(\mb{0},\mb{I}_K)$ are the information signal and the artificial noise, respectively. The non-negative diagonal matrix $\mb{\Psi}_t=\mathrm{diag}([\psi_{t,1},\ldots,\psi_{t,K}])$, $t\in\{s,a\}$, denotes the power allocation of the information $\mb{s}$ or the artificial noise $\mb{a}$, where the $k^{th}$ diagonal element $0\le \psi_{t,k}\le 1$ is the fraction of power allocated to the $k^{th}$ transmit symbol, and we have the power constraint: 
\begin{equation}
\mathrm{tr}(\mb{\Psi}_s + \mb{\Psi}_a)\le 1.\label{eqPowerCons}
\end{equation}
The precoding matrices $\mb{V}_s$ and $\mb{V}_a$ map the symbols to the transmit antennas. In particular, we consider the eigen-direction precoding~\cite{shafiee2007}, where $\mb{V}_s=\mb{V}_a=\mb{V}_1$ and the columns of $\mb{V}_1$ are the right singular vectors of the main channel $\mb{H}$, i.e., $\mb{H}$ has the singular value decomposition $\mb{H}=\mb{V}_0 \mb{\Lambda}^{1/2} \mb{V}_1^\dagger$. Although the precoding structure (\ref{eqx}) is heuristic and is generally not optimal due to the choice of $\mb{V}_s=\mb{V}_a=\mb{V}_1$, it is a reasonable scheme and leads to low-complexity transceiver design. First, we note that under the same channel knowledge assumption, the eigen-direction precoding achieves global optimal when Bob and Eve are equipped with single antenna~\cite{li2011}, and has been also adopted in other multi-antenna communication systems such as in~\cite{zhu2013} and~\cite{goel2008}. Second, by multiplying the receive signal $\mb{y}$ with the unitary matrix $\mb{V}_0^\dagger$, the received signals are decomposed into orthogonal parallel data streams with amplitudes proportional to $\lambda_k$, $1\le k\le \min(K,N)$, and therefore, simplifying the receiver design.

Under the precoding scheme (\ref{eqx}), $I(\mb{x};\mb{y},\mb{H})$ and $I(\mb{x};\mb{z}_i|\mb{F}_i)$ are given by the well-known mutual information of the Gaussian MIMO channels~\cite{telatar1999} as:
\begin{align}
&I(\mb{x};\mb{y},\mb{H}) = \log\frac{\det(\mb{I}+\gamma_B\mb{\Lambda}\mb{\Psi}_s + \gamma_B \mb{\Lambda}\mb{\Psi}_a)}{\det(\mb{I}+\gamma_B\mb{\Lambda}\mb{\Psi}_a)},\label{eqIxy}\\
&I(\mb{x};\mb{z}_i|\mb{F}_i) = \nonumber\\
&\hspace{-1ex} \mbb{E}\left[\log\frac{\det(\mb{I} + \gamma_E\mb{W}_i\mb{\Sigma}_i^{\frac{1}{2}}\mb{V}_1(\mb{\Psi}_s + \mb{\Psi}_a)\mb{V}_1^\dagger\mb{\Sigma}_i^{\frac{1}{2}}\mb{W}_i^\dagger)}{\det(\mb{I} + \gamma_E\mb{W}_i\mb{\Sigma}_i^{\frac{1}{2}}\mb{V}_1\mb{\Psi}_a\mb{V}_1^\dagger\mb{\Sigma}_i^{\frac{1}{2}}\mb{W}_i^\dagger)}\right],\label{eqIxz}
\end{align}
where $\gamma_B = P g_B/N_0$ and $\gamma_E = P g_E/N_0$. The expectation in (\ref{eqIxz}) is due to the definition of the conditional mutual information, and is taken over the distribution of $\mb{F}_i$. For the ease of exposition, denote the vectors $\bm{\gamma}_s=[\gamma_{s,1},\ldots,\gamma_{s,K}]$, $\bm{\gamma}_a=[\gamma_{a,1},\ldots,\gamma_{a,K}]$, where $\gamma_{s,k}=\gamma_E \psi_{s,k}$ and $\gamma_{a,k}=\gamma_E \psi_{a,k}$. Then, with $I(\mb{x};\mb{y},\mb{H})$ and $I(\mb{x};\mb{z}_i|\mb{F}_i)$ given in (\ref{eqIxy}) and (\ref{eqIxz}), a lower bound of the secrecy rate (\ref{eqRh_sec}), is expressed in terms of $\{\bm{\gamma}_s,\bm{\gamma}_a\}$ as:
\begin{equation}
\label{eqR_sec}
\begin{aligned}
R_{\mathrm{sec}}\!=\!\max_{\bm{\gamma}_s,\bm{\gamma}_a\in\mc{P}}\!\left[R_s(\bm{\gamma}_s,\bm{\gamma}_a)\right]^+\!,\ \mathrm{s.t.}\ \sum_{i=1}^K(\gamma_{s,i} + \gamma_{a,i})\le \gamma_E,
\end{aligned}
\end{equation}
where $\mc{P}$ denotes the set of $K$-dimensional non-negative real vectors. The constraint of (\ref{eqR_sec}) is rewritten from (\ref{eqPowerCons}) and $R_s(\bm{\gamma}_s,\bm{\gamma}_a)$ is given by:
\begin{align}
&R_s(\bm{\gamma}_s,\bm{\gamma}_a) = f_B(\bm{\gamma}_s + \bm{\gamma}_a) - f_B(\bm{\gamma}_a)\nonumber\\
&\qquad-\max_{1\le i\le L}\left\{f_{E,i}(\bm{\gamma}_s + \bm{\gamma}_a;\mb{V}_1) - f_{E,i}(\bm{\gamma}_a;\mb{V}_1)\right\},\label{eqRs}
\end{align}
where $\bm{\gamma}_s+\bm{\gamma}_a$ is defined as the element-wise addition, functions $f_B(\bm{\gamma})$ and $f_{E,i} (\bm{\gamma};\mb{V}_1)$ are:
\begin{equation}
f_B(\bm{\gamma}) = \log\det\left(\mb{I} + \frac{\gamma_B}{\gamma_E}\mb{\Lambda}\mb{\Gamma}\right),\label{eqfb}
\end{equation}
\begin{equation}
f_{E,i}(\bm{\gamma};\mb{V}_1) = \mbb{E}\left[\log\det\left(\mb{I} + \mb{W}_i\mb{\Sigma}_i^{\frac{1}{2}}\mb{V}_1\mb{\Gamma}\mb{V}_1^\dagger\mb{\Sigma}_i^{\frac{1}{2}}\mb{W}_i^{\dagger}\right)\right],\label{eqfe}
\end{equation}
and $\mb{\Gamma}=\mathrm{diag}(\bm{\gamma})$. Note that $f_{E,i} (\bm{\gamma};\mb{V}_1)$ depends on the channel $\mb{H}$ via the precoding matrix $\mb{V}_1$, which is fixed under the expectation operation. In Section~\ref{secAvgRate}, we will derive approximations for $f_{E,i}(\cdot)$ under different channel settings, and in Section~\ref{secOpt} present an optimization framework to obtain the optimal power allocation as to maximize $R_{\mathrm{sec}}$ in (\ref{eqR_sec}).

\section{Average Rate of Eigen-Direction Precoding}\label{secAvgRate}
In this section, we present analytical approximations for the function $f_{E,i}(\bm{\gamma};\mb{V}_1)$, defined in (\ref{eqfe}). Compared to the exact expression of $f_{E,i}(\bm{\gamma};\mb{V}_1)$, the proposed approximations are explicit functions of the variable $\bm{\gamma}$, which can be conveniently used to deduce the corresponding linear affine approximations of $f_{E,i}(\bm{\gamma};\mb{V}_1)$. As will be shown in Section~\ref{secOpt}, the affine approximation of $f_{E,i}(\bm{\gamma};\mb{V}_1)$ is the key ingredient to convert the non-convex problem (\ref{eqR_sec}) into a sequence of simple convex sub-problems. Under some typical system settings of the cooperative MIMO, we also present examples of numerical results to illustrate the approximation error incurred by using the proposed approximations. For notational simplicity, whenever it is clear from the context, we drop the dependency of the location index $i$ from $f_{E,i}(\bm{\gamma};\mb{V}_1)$ and its approximations.

\subsection{Average Rate Between Alice and Eve}

Given a certain power allocation $\bm{\gamma}$ and a precoding $\mb{V}_1$, $f_E(\bm{\gamma};\mb{V}_1)$ in (\ref{eqfe}) has the same formulation as the correlated Rayleigh MIMO channel with the transmitter-side correlation, which is available in literature~\cite{simon2006}. However, by using the expression~\cite[Eq. (123)]{simon2006}, $f_E(\bm{\gamma};\mb{V}_1)$ depends on the power allocation $\bm{\gamma}$ via the eigenvalues of $\mb{\Sigma}^{1/2}\mb{V}_1\mb{\Gamma}\mb{V}_1^\dagger\mb{\Sigma}^{1/2}$, instead of $\bm{\gamma}$ itself. As a result, it is inconvenient to use $f_E(\bm{\gamma};\mb{V}_1)$ in the power optimization, as it incurs large computation complexity required to solve the eigenvalue problems. In addition, the implicit solution of the function $f_E(\bm{\gamma};\mb{V}_1)$ would prevent further manipulations to simplify the secrecy rate optimization procedures, such as those described in Section~\ref{secOpt}.

To address this issue, we propose approximations for $f_E(\bm{\gamma};\mb{V}_1)$, which will directly depend on $\bm{\gamma}$. Specifically, consider the quantity $\widetilde{f_E}(\bm{\gamma})$ constructed as follows:
\begin{equation}
\widetilde{f_E}(\bm{\gamma}) = \log\mbb{E}[\det(\mb{I} + \mb{W}\mb{\Sigma}^{1/2}\mb{U}\mb{\Gamma}\mb{U}^\dagger\mb{\Sigma}^{1/2}\mb{W}^\dagger)],\label{eqfte}
\end{equation}
where the expectation is taken over both $\mb{W}$ and the Haar unitary matrix $\mb{U}\in\mc{U}_K$. We denote $\mc{U}_K$ as the unitary group containing all $K\times K$ unitary matrices~\cite{sterngerg1995}. Comparing (\ref{eqfte}) with (\ref{eqfe}), we replace the fixed unitary matrix $\mb{V}_1$ with the random unitary matrix $\mb{U}$ and apply the Jensen's inequality for the concave $\log(\cdot)$ function. Intuitively, when the path-loss matrix $\mb{\Sigma}$ is close to an identity matrix, the unitary matrix $\mb{V}_1$ or $\mb{U}$ commutes with $\mb{\Sigma}$ and can be absorbed into the unitary invariant matrix $\mb{W}$. The quantity $\widetilde{f_E}(\bm{\gamma})$ then becomes the strict upper bound of $f_E(\bm{\gamma};\mb{V}_1)$. In addition, under other typical settings of the cooperative MIMO systems, we will also show that $\widetilde{f_E}(\bm{\gamma})$ can, in fact, serve as an accurate approximation for $f_E(\bm{\gamma};\mb{V}_1)$.

In the following, we will denote $\{\mc{D}_{i,j}\}_{1\le i\le a,1\le j\le b}$ as a $a\times b$ matrix block. We denote $\Delta_m(\bm{a})=\det[a_i^{j-1}] = \prod_{1\le j\le k\le m}(a_k-a_j)$ as the Vandermonde determinant and denote ${}_p F_q\left(\begin{array}{c} a_1,\ldots,a_p \\ b_1,\ldots,b_q \end{array}\middle | x\right)$ as the generalized hypergeometric function with $p+q$ parameters~\cite[Eq. (9.14.1)]{gradshteyn2014}. When $|x|<1$, it admits the series representation:
\begin{equation}\label{eqHyper}
{}_p F_q\left(\begin{array}{c} a_1,\ldots,a_p \\ b_1,\ldots,b_q \end{array}\middle | x\right) = \sum_{n=0}^\infty \frac{(a_1)_n\ldots (a_p)_n}{(b_1)_n\ldots (b_q)_n}\frac{x^n}{n!},
\end{equation}
where $(a)_n=\Gamma(a+n)/\Gamma(a)$ is the Pochhammer symbol. In the following propositions, without loss of generality, we present the expressions of $\widetilde{f_E}(\bm{\gamma})$ when $\gamma_1\ge\cdots\ge\gamma_K$. Indeed, from (\ref{eqfte}), it is easy to check $\widetilde{f_E}(\bm{\gamma})$ is invariant under any permutation of the elements of $\bm{\gamma}$, since $\mb{U}\mb{\Gamma}\mb{U}^\dagger$ has the same distribution as $\mb{U}\mb{P}\mb{\Gamma}\mb{P}^\dagger\mb{U}^\dagger$ for any permutation matrix $\mb{P}$. We recall $\{\sigma_k\}_{1\le k\le K}$ are the average channel gains between Alice and Eve as defined in (\ref{eqsigma}).

\begin{proposition}\label{prop1}
If $M\ge K\ge n$, $\gamma_1\ge\cdots\ge\gamma_n>0$ and $\gamma_{n+1}=\cdots=\gamma_K=0$, then
\begin{align}
\widetilde{f_E}(\bm{\gamma}) = &\sum_{j=K-n}^{K-1} \log\frac{\Gamma(K+1-j)\Gamma(j+1)\Gamma(M-K+1)}{\Gamma(K+1)\Gamma(M-K+j+1)}\nonumber\\
&\qquad+\log\left(\frac{\det[\mc{A}]}{\Delta_K(\bm{\sigma})\Delta_n(\bm{\gamma})\prod_{i=1}^n\gamma_i^{K-n}}\right),\label{eqfe_tilde1}
\end{align}
where the $K\times K$ matrix $\mc{A}$ is given by
\begin{equation}\label{eqA}
\hspace{-1ex}\mc{A}\!=\!\left[\!\begin{array}{c}
\{\sigma_j^{i-1}\}_{1\le i\le K-n, 1\le j\le K} \\
 \left\{{}_2 F_0\left(\begin{array}{c} M-K+1,-K\\ - \end{array}\middle|-\sigma_j\gamma_i\right)\right\}_{\substack{1\le i\le n \\ 1\le j\le K}}
\end{array}\right].
\end{equation}
\end{proposition}
\begin{proof}
The proof of Proposition~\ref{prop1} is in Appendix~\ref{appx1}.
\end{proof}

\begin{proposition}\label{prop2}
If $K>M\ge n$, $\gamma_1\ge\cdots\ge\gamma_n>0$ and $\gamma_{n+1}=\cdots=\gamma_K=0$, then
\begin{align}
\widetilde{f_E}(\bm{\gamma}) = & \log\frac{\det[\mc{B}]\prod_{i=1}^n\gamma_i^{n-M}}{\Delta_n(\bm{\gamma})\Delta_K(\bm{\sigma})} - n\log\Gamma(K-M+1)\nonumber\\
&+\sum_{j=K-n}^{K-1}\log\frac{\Gamma(K+1-j)\Gamma(j+1)}{\Gamma(M+1)\Gamma(M-K+j+1)},\label{eqfe_tilde2}
\end{align}
where the $K\times K$ matrix $\mc{B}$ is given by
\begin{equation}\label{eqB}
\hspace{-1.5ex}\mc{B}\!=\!\left[\!\begin{array}{c}
\!\{\sigma_j^{i-1}\}_{1\le i\le K-n, 1\le j\le K} \\
\!\left\{\sigma_{j}^{K-M}{}_3 F_1\left(\begin{array}{c} 1,1,-M\\ K-M+1 \end{array}\middle|-\sigma_j\gamma_i\right)\right\}_{\substack{1\le i\le n \\ 1\le j\le K}}\!
\end{array}\!\right]\!.
\end{equation}
If $K>n>M$, then
\begin{align}
&\widetilde{f_E}(\bm{\gamma}) = \log\frac{(-1)^{(K-M)(n-M)}\det[\mc{C}]}{\Delta_n(\bm{\gamma})\Delta_K(\bm{\sigma})} - M\log\Gamma(M+1)\nonumber\\
&\qquad+\sum_{j=n-M}^{n-1}\log\frac{\Gamma(n+1-j)\Gamma(K-n+j+1)}{\Gamma(K-M+1)\Gamma(M-n+j+1)},\label{eqfe_tilde3}
\end{align}
where the $(K+n-M)\times (K+n-M)$ matrix $\mc{C}$ is given by (\ref{eqC}) on top of the next page.
\end{proposition}
\begin{proof}
The proof of Proposition~\ref{prop2} is in Appendix~\ref{appx2}.
\end{proof}

\begin{figure*}[!t]
	\normalsize
	\setcounter{equation}{21}
	
	\begin{equation}
		\mc{C} = \left[\begin{array}{cc}
		\{0\}_{1\le i\le K-M, 1\le j\le n-M} & \{\sigma_j^{i-1}\}_{1\le i\le K-M, 1\le j\le K} \\
 		\{\gamma_{i}^{j-1}\}_{1\le i\le n, 1\le j\le n-M} & \left\{\gamma_{i}^{n-M}\sigma_{j}^{K-M}{}_3 F_1\left(\begin{array}{c} 1,1,-M\\ K-M+1 \end{array}\middle|-\sigma_j\gamma_i\right)	\right\}_{1\le i\le n, 1\le j\le K}
		\end{array}\right].\label{eqC}
	\end{equation}
	
	\setcounter{equation}{\value{MYtempeqncnt}}
	\hrulefill
	\vspace*{4pt}
\end{figure*}
\setcounter{equation}{22}

According to the definition (\ref{eqHyper}), the generalized hypergeometric functions in (\ref{eqA}), (\ref{eqB}), and (\ref{eqC}) reduce to finite series summations given as:
\begin{align*}
&{}_2 F_0\left(\begin{array}{c} M-K+1,-K\\ - \end{array}\middle|-x\right)\nonumber\\
&\qquad = \sum_{l=0}^K \frac{\Gamma(M-K+l+1)\Gamma(K+1)x^l}{\Gamma(M-K+1)\Gamma(K+1-l)\Gamma(l+1)},\\
&{}_3 F_1\left(\begin{array}{c} 1,1,-M\\ K-M+1 \end{array}\middle|-x\right)\nonumber\\
&\qquad = \sum_{l=0}^M \frac{\Gamma(l+1)\Gamma(M+1)\Gamma(K-M+1)x^l}{\Gamma(K-M+l+1)\Gamma(M+1-l)}.
\end{align*}
The Propositions~\ref{prop1} and~\ref{prop2} can be used to compute $\widetilde{f_E}(\bm{\gamma})$ when $\bm{\gamma}$ is rank deficient, i.e., $n<K$. The rank deficiency of $\bm{\gamma}$ may be due to the power optimization process, when Alice does not allocate power to certain eigen-channels, so as to reduce the information rate towards Eve.

\subsection{Examples of Numerical Results}

In this subsection, we present numerical results to validate the approximate expressions derived in Propositions~\ref{prop1} and~\ref{prop2}. Note that in this section, we do not optimize the power allocation; this is the subject of the next section. We denote $R_E=f_E(\bm{\gamma}_s+\bm{\gamma}_a;\mb{V}_1)-f_E(\bm{\gamma}_a;\mb{V}_1)$ as the information rate between Alice and Eve, as per (\ref{eqRs}), and denote $\widetilde{R_E}=\widetilde{f_E}(\bm{\gamma}_s+\bm{\gamma}_a)-\widetilde{f_E}(\bm{\gamma}_a)$ as the approximation of $R_E$, where $\widetilde{f_E}(\cdot)$ is calculated by (\ref{eqfe_tilde1}) when $K\le M$, and by (\ref{eqfe_tilde2}) and (\ref{eqfe_tilde3}) when $K>M$. In specific, we illustrate the variation range of $R_E$ in Fig.~\ref{figRE} (a) and (b) as the vertical bars, i.e., the length of each bar is $\max R_E -\min R_E$, where the variation of $R_E$ is due to $10^4$ samples of the precoding matrices $\mb{V}_1$, a.k.a., $10^4$ samples of the main channel $\mb{H}$. Recall that the information rate $R_E$ is calculated by averaging the eavesdropping channel coefficients $\mb{W}$ only, while the precoding $\mb{V}_1$ is fixed. The values of $R_E$ is then compared with the corresponding $\widetilde{R_E}$. In Fig.~\ref{figRE} (c) and (d), we plot the Cumulative Distribution Function (CDF) of the approximation error $\widetilde{R_E}-R_E$, which is induced by replacing the fixed $\mb{V}_1$ in $R_E$ with the random Haar unitary matrix $\mb{U}$ in $\widetilde{R_E}$. When forming the cooperative MIMO system, the transmit nodes are randomly distributed in a circular area with the radius $r_T$ from 1 to 9 meters. The distance between the transmit head node and Eve is set to $d_E=30$ meters. The number of receive antennas at Eve is assumed to be $M=2$, 4, or 6, while the number of transmitters at Alice is set to $K=2$ or 6. For each $K$, the power allocation vectors $\bm{\gamma}_s$ and $\bm{\gamma}_a$ are randomly chosen and then fixed for random samples of the fast fading $\mb{W}$.

When $r_T=1$, the transmit nodes are distributed within a small area and their path losses towards Eve have similar value. The corresponding CDF curves in Fig.~\ref{figRE} (c) and (d) show that the approximation error $\widetilde{R_E}-R_E$ can be bounded within 0.4 nats/s/Hz. As the cluster radius increases, the transmit nodes are distributed in a larger area and the values of the elements of $\mb{\Sigma}$ are more dispersive. When $r_T$ increases to 9 meters, the approximation error can be up-to 0.4 nats/s/Hz when there are 2 transmit nodes, and up-to 0.7 nats/s/Hz when there are 6 transmit nodes. Overall, the approximation $\widetilde{R_E}$ of the rate $R_E$ is reasonably accurate, especially when $K$ and $M$ are large. In all the considered cases, the approximation $\widetilde{R_E}$ over-estimates the corresponding $R_E$ with a large probability, i.e., as seen in Fig.~\ref{figRE}, the error does not take negative values with appreciable probability. Therefore, in the context of the physical-layer security, it is relevant to use $\widetilde{R_E}$ as an approximation of $R_E$ in (\ref{eqRs}) as it yields the lower bound of the achievable secrecy rate. Designing a system by using such a lower bound to configure the information rate of the cooperative MIMO transmissions would not violate the system's secrecy requirements.

\section{Transmit Power Optimization}\label{secOpt}

\subsection{Iterative Power Optimization}\label{secOptAlg}

The secrecy rate maximization (\ref{eqR_sec}) is non-convex and its global optimal solution is hard to obtain in general. To address this issue, we propose to solve the problem (\ref{eqR_sec}) via the successive convex approximations. Specifically, we first approximate the optimizing function $R_s(\bm{\gamma}_s,\bm{\gamma}_a)$ with $\widetilde{R_s}(\bm{\gamma}_s,\bm{\gamma}_a)$ by replacing $f_E(\bm{\gamma};\mb{V}_1)$ in (\ref{eqRs}) with its closed-form approximations $\widetilde{f_E}(\bm{\gamma})$ given in Propositions~\ref{prop1} and~\ref{prop2}. We further approximate the non-convex function $\widetilde{R_s}(\bm{\gamma}_s,\bm{\gamma}_a)$ with a concave function by taking the affine Taylor expansions of the $f_B(\cdot)$ and $\widetilde{f_E}(\cdot)$ components. That is, in a neighborhood of $\bm{\gamma}_s^{(0)},\bm{\gamma}_a^{(0)}\in\mc{P}$, $\widetilde{R_s}(\bm{\gamma}_s,\bm{\gamma}_a)$ can be approximated as:
\begin{align}
&\quad \widetilde{R_s} \left(\bm{\gamma}_s,\bm{\gamma}_a\middle|\bm{\gamma}_s^{(0)},\bm{\gamma}_a^{(0)}\right)\nonumber\\
& = f_B(\bm{\gamma}_s + \bm{\gamma}_a) - f_B\left(\bm{\gamma}_a^{(0)}\right) - \nabla_B\left(\bm{\gamma}_a^{(0)}\right)^{\mathrm{T}}\left(\bm{\gamma}_a - \bm{\gamma}_a^{(0)}\right)\nonumber\\
& \qquad - \max_{1\le i\le L}\left\{\widetilde{f_{E,i}}\left(\bm{\gamma}_s^{(0)}+\bm{\gamma}_a^{(0)}\right) - \widetilde{f_{E,i}}\left(\bm{\gamma}_a^{(0)}\right) \right.\nonumber\\
& \qquad + \nabla_{E,i}\left(\bm{\gamma}_s^{(0)} + \bm{\gamma}_a^{(0)}\right)^{\mathrm{T}}\left(\bm{\gamma}_s+\bm{\gamma}_a - \bm{\gamma}_s^{(0)} - \bm{\gamma}_a^{(0)}\right)\nonumber\\
& \qquad -\nabla_{E,i}\left(\bm{\gamma}_a^{(0)}\right)^{\mathrm{T}}\left(\bm{\gamma}_a - \bm{\gamma}_a^{(0)}\right)\Big\},\label{eqRs_tilde}
\end{align}
where $\nabla_B(\bm{\gamma}^{(0)})=\left[\frac{\partial}{\partial \gamma_1} f_B(\bm{\gamma}),\ldots,\frac{\partial}{\partial \gamma_K} f_B(\bm{\gamma})\right]_{\bm{\gamma}=\bm{\gamma}^{(0)}}^\mathrm{T}$, $\nabla_{E,i}(\bm{\gamma}^{(0)})=\left[\frac{\partial}{\partial \gamma_1} \widetilde{f_{E,i}}(\bm{\gamma}),\ldots,\frac{\partial}{\partial \gamma_K} \widetilde{f_{E,i}}(\bm{\gamma})\right]_{\bm{\gamma}=\bm{\gamma}^{(0)}}^\mathrm{T}$ are the gradients of $f_B(\bm{\gamma})$ and $\widetilde{f_{E,i}}(\bm{\gamma})$ evaluated at $\bm{\gamma}^{(0)}$. As the first term of (\ref{eqRs_tilde}) is a concave function of $\{\bm{\gamma}_s,\bm{\gamma}_a\}$ and the other terms are linear, $\widetilde{R_s}\left(\bm{\gamma}_s,\bm{\gamma}_a\middle | \bm{\gamma}_s^{(0)},\bm{\gamma}_a^{(0)}\right)$ is a concave function. Together with the linear constraint $\sum_{i=1}^K(\gamma_{s,i}+\gamma_{a,i})\le \gamma_E$ in (\ref{eqR_sec}), the maximization of the function $\widetilde{R_s}\left(\bm{\gamma}_s,\bm{\gamma}_a \middle |\bm{\gamma}_s^{(0)},\bm{\gamma}_a^{(0)}\right)$, in the neighborhood of $\{\bm{\gamma}_s^{(0)},\bm{\gamma}_a^{(0)}\}$, can be efficiently solved by convex programming tools, such as CVX~\cite{grant2013,grant2008}.

Next, we derive the expressions of the gradients $\nabla_B$ and $\nabla_{E,i}$, as needed in (\ref{eqRs_tilde}). Based on (\ref{eqfb}), the partial derivative $\frac{\partial}{\partial \gamma_i} f_B(\bm{\gamma})$ is given by:
\begin{equation}
\frac{\partial}{\partial \gamma_i} f_B(\bm{\gamma}) = \frac{\gamma_B \lambda_i}{\gamma_E + \gamma_B\lambda_i\gamma_i}.
\end{equation}
Denote $\mb{V}(\bm{\gamma})=\{\gamma_i^{j-1}\}_{1\le i\le n,1\le j\le n}$ as the Vandermonde matrix. When $M\ge K$, $\widetilde{f_{E,i}}(\bm{\gamma})$ is given by (\ref{eqfe_tilde1}) and the derivative $\frac{\partial}{\partial \gamma_i} \widetilde{f_{E,i}}(\bm{\gamma})$ can be obtained by the Jacobi's formula~\cite{magnus1999} as:
\begin{align}
\frac{\partial}{\partial \gamma_i} \widetilde{f_{E,i}} &(\bm{\gamma}) = \mathrm{tr}\left(\mc{A}^{-1}\frac{\partial}{\partial\gamma_i}\mc{A}\right)\nonumber\\
& - \mathrm{tr}\left(\mb{V}(\bm{\gamma})^{-1}\frac{\partial}{\partial\gamma_i}\mb{V}(\bm{\gamma})\right) - \frac{K-n}{\gamma_i},\label{eqDfe_tilde}
\end{align}
where the partial derivative of a $a\times a$ matrix $\mb{A}(t)$ with respect to the variable $t$ is defined as $\frac{\partial}{\partial t} \mb{A}(t)=\left\{\frac{\partial}{\partial t} \mb{A}_{i,j}(t)\right\}_{1\le i\le a,1\le j\le a}$. Inserting (\ref{eqA}) into (\ref{eqDfe_tilde}) and applying Laplace expansion of the determinants, we obtain: 
\begin{align}
\frac{\partial}{\partial \gamma_i} \widetilde{f_{E,i}}(\bm{\gamma}) & =  -\frac{1}{\det[\mb{V}(\bm{\gamma})]}\sum_{j=1}^n(-1)^{i+j}M_{j,i}^\mb{V}\gamma_i^{j-2}-\frac{K-n}{\gamma_i} \nonumber\\
&+\frac{K(M-K+1)}{\det[\mc{A}]} \sum_{j=1}^K \frac{\sigma_j M_{j,K-n+i}^{\mc{A}}}{(-1)^{K-n+i+j}}\nonumber\\
&\times {}_2F_0\left( \begin{array}{c} M-K+2,1-K\\ - \end{array} \middle|-\sigma_j\gamma_i\right),\label{eqDfe_tilde_b}
\end{align}
where $M_{a,b}^\mc{A}$ and $M_{a,b}^\mb{V}$ are the $(a,b)$ minors of the matrices $\mc{A}$ and $\mb{V}(\bm{\gamma})$. In (\ref{eqDfe_tilde_b}), we applied the differentiation formula~\cite[Eq. (16.3.1)]{olver2010} for the generalized hypergeometric function.

Similarly, when $K>M\ge n$, the derivative $\frac{\partial}{\partial \gamma_i} \widetilde{f_{E,i}}(\bm{\gamma})$ can be obtained by using (\ref{eqfe_tilde2}) as:
\begin{align*}
&\frac{\partial}{\partial \gamma_i} \widetilde{f_{E,i}}(\bm{\gamma})\nonumber\\
& = \mathrm{tr}\left(\mc{B}^{-1}\frac{\partial}{\partial\gamma_i}\mc{B}\right) - \mathrm{tr}\left(\mb{V}(\bm{\gamma})^{-1}\frac{\partial}{\partial\gamma_i}\mb{V}(\bm{\gamma})\right) - \frac{M-n}{\gamma_i}\nonumber\\
& = -\frac{1}{\det[\mb{V}(\bm{\gamma})]}\sum_{j=1}^n(-1)^{i+j}M_{j,i}^\mb{V}\gamma_i^{j-2}-\frac{M-n}{\gamma_i} \nonumber\\
&\qquad +\frac{M}{(K-M+1)\det[\mc{B}]} \sum_{j=1}^K \frac{\sigma_j^{K-M+1} M_{j,K-n+i}^{\mc{B}}}{(-1)^{K-n+i+j}}\nonumber\\
&\qquad \times {}_3F_1\left( \begin{array}{c} 2,2,1-M\\ K-M+2 \end{array} \middle|-\sigma_j\gamma_i\right),
\end{align*}
When $K>n>M$, the derivative $\frac{\partial}{\partial \gamma_i} \widetilde{f_{E,i}}(\bm{\gamma})$ is obtained by using (\ref{eqfe_tilde3}) and is shown on top of the next page.

\begin{figure*}[!t]
	\normalsize
	\setcounter{equation}{26}
	
\begin{align*}
&\frac{\partial}{\partial \gamma_i} \widetilde{f_{E,i}}(\bm{\gamma}) = \mathrm{tr}\left(\mc{C}^{-1}\frac{\partial}{\partial\gamma_i}\mc{C}\right) - \mathrm{tr}\left(\mb{V}(\bm{\gamma})^{-1}\frac{\partial}{\partial\gamma_i}\mb{V}(\bm{\gamma})\right) \nonumber\\
& = \frac{1}{\det[\mc{C}]}\sum_{j=1}^{n-M}\frac{(j-1)M_{j,K-M+i}^{\mc{C}}}{(-1)^{K-M+i+j}}\gamma_i^{j-2} -\frac{1}{\det[\mb{V}(\bm{\gamma})]}\sum_{j=1}^n(-1)^{i+j}M_{j,i}^\mb{V}\gamma_i^{j-2} \nonumber\\
&\qquad +\frac{n-M}{\det[\mc{C}]} \sum_{j=1}^K \frac{\gamma_i^{n-M-1}\sigma_j^{K-M} M_{n-M+j,K-M+i}^{\mc{C}}}{(-1)^{K-2M+n+i+j}} {}_3F_1\left( \begin{array}{c} 1,1,-M\\ K-M+1 \end{array} \middle|-\sigma_j\gamma_i\right)\nonumber\\
&\qquad +\frac{M}{(K-M+1)\det[\mc{C}]} \sum_{j=1}^K \frac{\gamma_i^{n-M}\sigma_j^{K-M+1} M_{n-M+j,K-M+i}^{\mc{C}}}{(-1)^{K-2M+n+i+j}} {}_3F_1\left( \begin{array}{c} 2,2,1-M\\ K-M+2 \end{array} \middle|-\sigma_j\gamma_i\right).
\end{align*}
	
	\setcounter{equation}{\value{MYtempeqncnt}}
	\hrulefill
	\vspace*{4pt}
\end{figure*}
\setcounter{equation}{26}

\begin{figure*}[t!]
\centering
\subfigure[]{\includegraphics[width=2.8in]{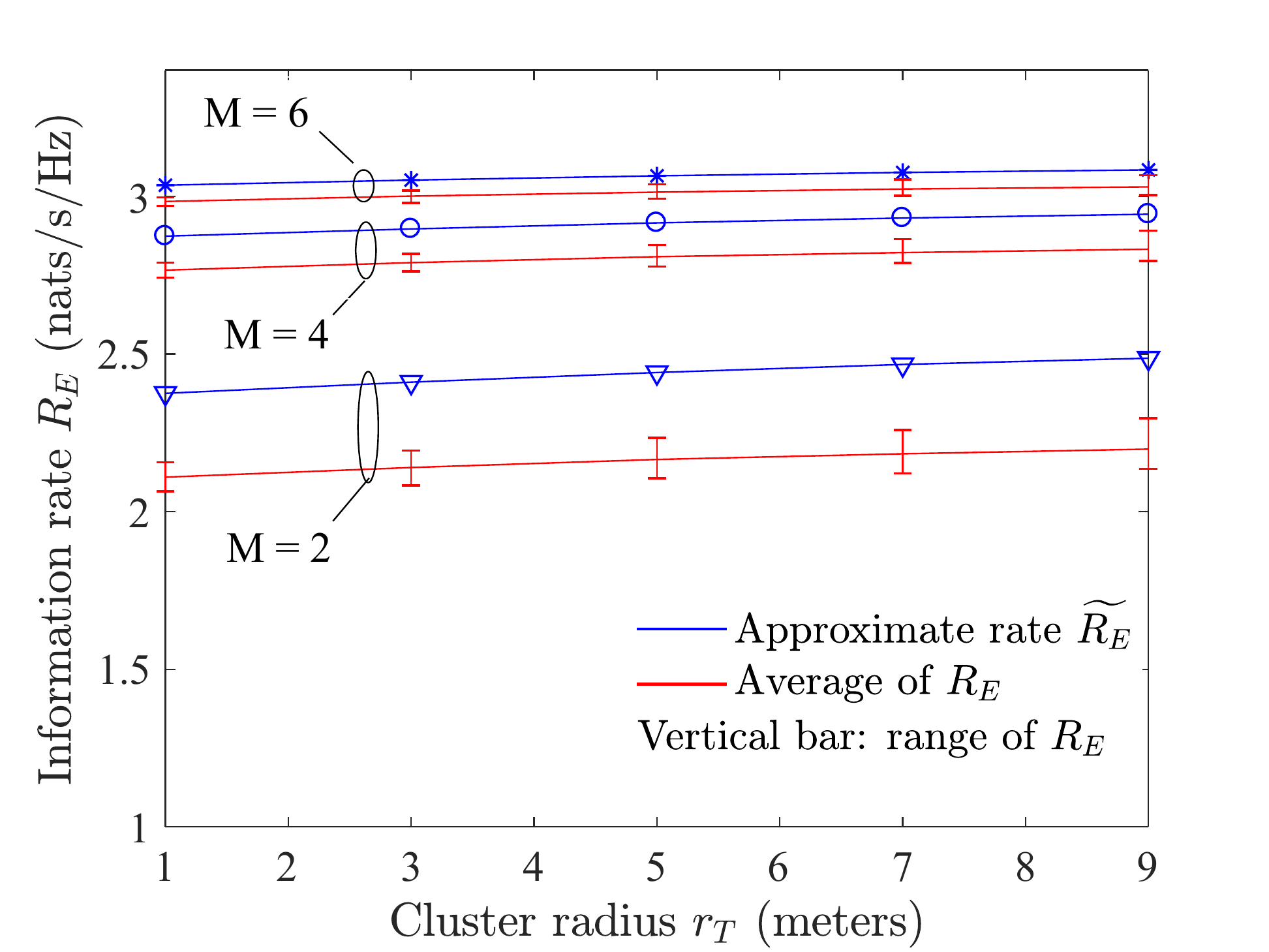}
}
\subfigure[]{\includegraphics[width=2.8in]{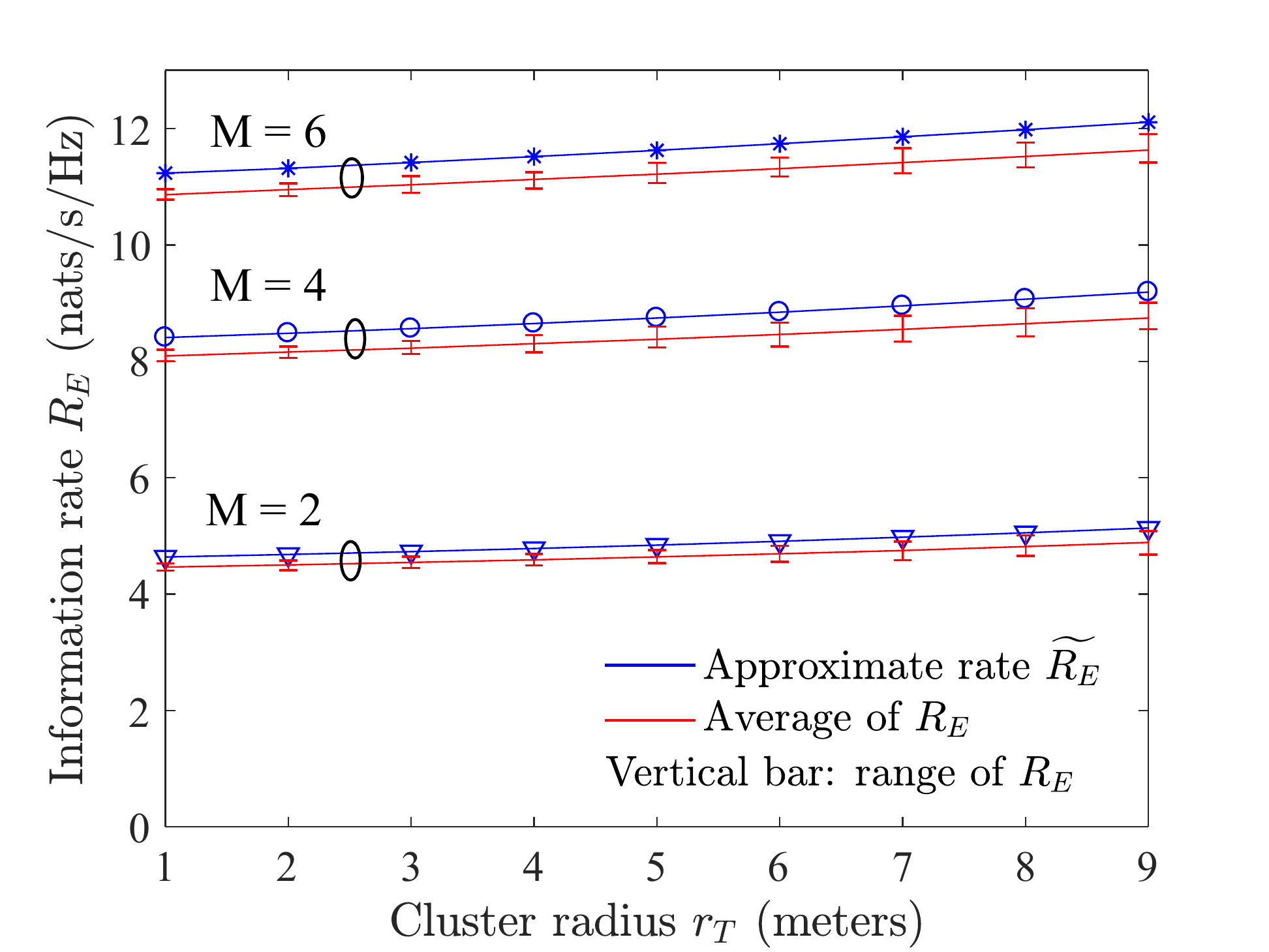}
}\\
\subfigure[]{\includegraphics[width=2.8in]{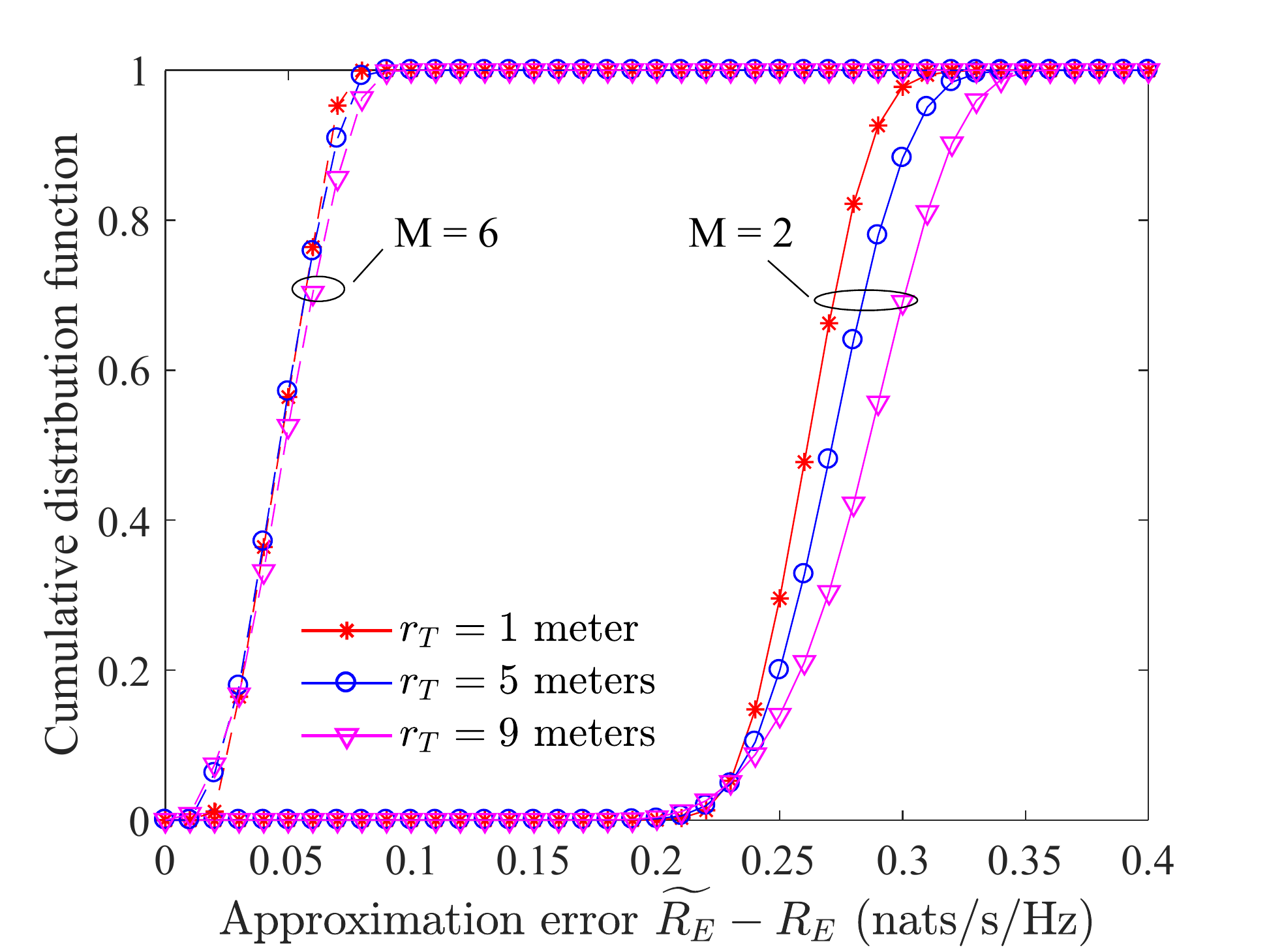}
}
\subfigure[]{\includegraphics[width=2.8in]{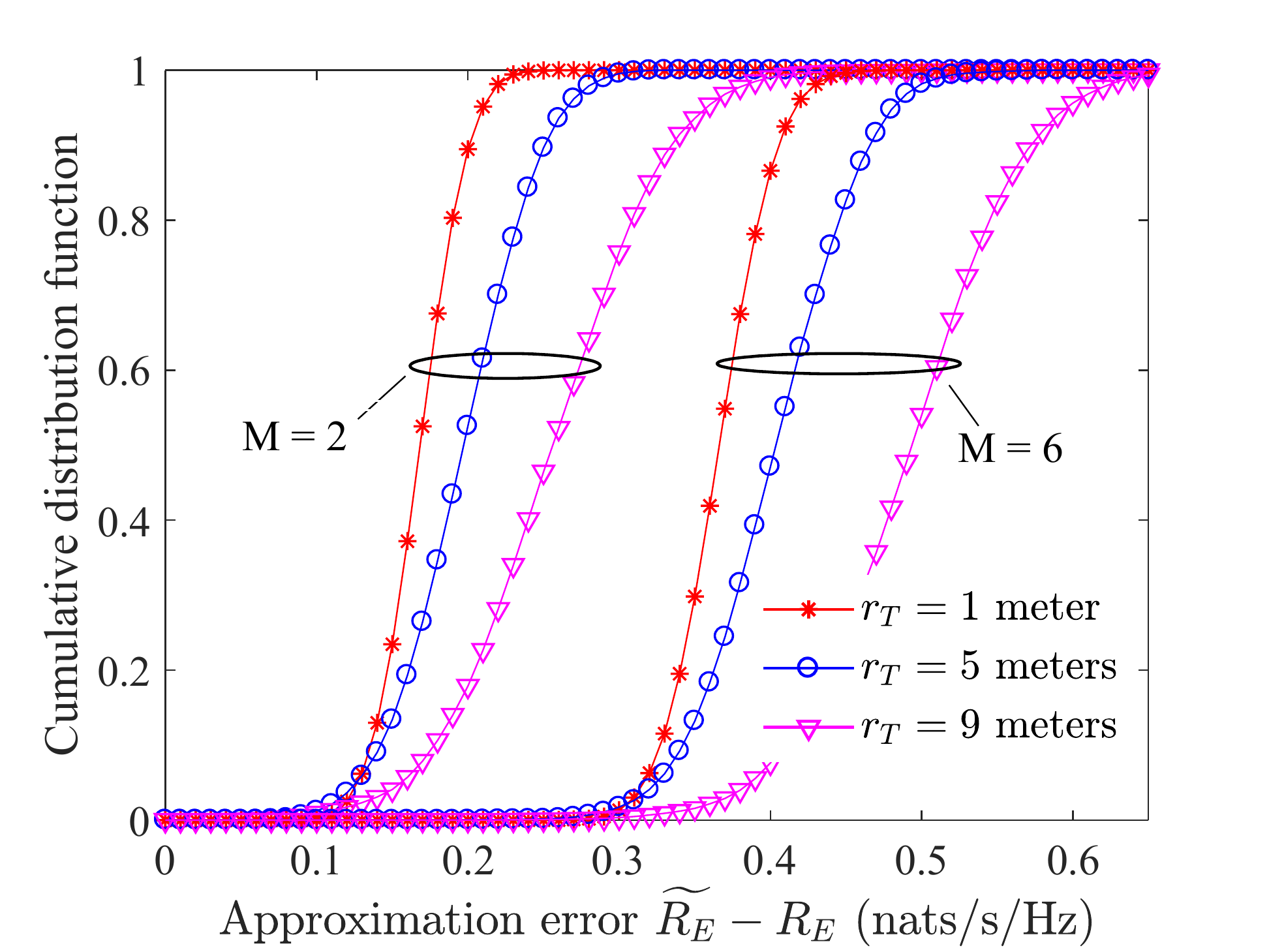}
}
\caption{Comparisons between $R_E$ and $\widetilde{R_E}$. The distance between the head of Alice and Eve is $d_E=30$ meters and the radius of Alice is from $r_T=1$ to 9 meters. The locations of the transmit nodes and the power allocation vectors $\bm{\gamma}_s$ and $\bm{\gamma}_a$ are random generated. Vertical bars in (a) and (b) denote the range of $R_E$ due to different precoding matrix $\mb{V}_1$. $\widetilde{R_E}$ and range of $R_E$ with $K=2$ in (a) and with $K=6$ in (b); CDFs of approximation error $\widetilde{R_E}-R_E$ with $K=2$ in (c) and with $K=6$ in (d).}
\label{figRE}
\end{figure*}

Finally, following similar procedures as described in~\cite{cumanan2014}, the non-convex secrecy rate maximization problem (\ref{eqR_sec}) can be converted into a sequence of concave maximization sub-problems. Specifically, selecting initial points $\bm{\gamma}_s^{(0)},\bm{\gamma}_a^{(0)}\in\mc{P}$, the $k^{th}$ sub-problem is given as:
\begin{align}
\left\{\bm{\gamma}_s^{(k)},\bm{\gamma}_a^{(k)}\right\} = &
\underset{\bm{\gamma}_s,\bm{\gamma}_a\in\mc{P}}{\arg\max}\left[\widetilde{R_s}\left(\bm{\gamma}_s,\bm{\gamma}_a \middle| \bm{\gamma}_s^{(k-1)},\bm{\gamma}_a^{(k-1)}\right)\right]^+,\nonumber\\
& \mathrm{s.t.}\ \sum_{i=1}^K(\gamma_{s,i} + \gamma_{a,i})\le \gamma_E,\label{eqR_sec_tilde}
\end{align}
where $\{\bm{\gamma}_s^{(k)},\bm{\gamma}_a^{(k)}\}$ denotes the optimal solution of the $k^{th}$ sub-problem. The initial signal power allocation $\bm{\gamma}_s^{(0)}$ is selected as the solution of the water-filling algorithm, which only maximizes the information rate between Alice and Bob, i.e., $\bm{\gamma}_{\mathrm{WF}}=\arg\max\{f_B(\bm{\gamma}): \bm{\gamma}\in\mc{P},\sum_{i=1}^K \gamma_i \le \gamma_E \}$. The initial power allocation of the artificial noise $\bm{\gamma}_a^{(0)}$ is chosen as an all-zero vector. The iterative power optimization procedure is summarized as Algorithm~\ref{alg1}. As shown in line 5 of Algorithm~\ref{alg1}, the algorithm terminates when the secrecy rate improvement between consecutive sub-problems is less than a threshold $\epsilon$.

\begin{algorithm}[b]
	\caption{Iterative Power Optimization}\label{alg1}
	\begin{algorithmic}[1]
		\State \textbf{initialize:} $k=1$, $\epsilon>0$, $\bm{\gamma}_s^{(0)} = \bm{\gamma}_{\mathrm{WF}}$, and $\bm{\gamma}_a^{(0)} = \mb{0}$
		\State \textbf{repeat}
		\State\qquad Solve the sub-problem (\ref{eqR_sec_tilde}) with CVX and set the output as $\left\{\bm{\gamma}_s^{(k)}, \bm{\gamma}_a^{(k)}\right\}$;
		\State\qquad $k = k+1$;
		\State \textbf{until} $\widetilde{R_s}\left(\bm{\gamma}_s^{(k+1)},\bm{\gamma}_a^{(k+1)}\middle| \bm{\gamma}_s^{(k)},\bm{\gamma}_a^{(k)}\right) - \widetilde{R_s}\left(\bm{\gamma}_s^{(k)},\bm{\gamma}_a^{(k)}\middle| \bm{\gamma}_s^{(k-1)},\bm{\gamma}_a^{(k-1)}\right)\le \epsilon$ 
	\end{algorithmic}
\end{algorithm}

Consider the consecutive $(k+1)^{st}$ and $k^{th}$ sub-problems, where $\{\bm{\gamma}_s^{(k+1)},\bm{\gamma}_a^{(k+1)}\}$, $\{\bm{\gamma}_s^{(k)},\bm{\gamma}_a^{(k)}\}$, and $\{\bm{\gamma}_s^{(k-1)},\bm{\gamma}_a^{(k-1)}\}$ are all feasible solutions. Based on Algorithm~\ref{alg1}, we have the inequalities $\widetilde{R_s}\left(\bm{\gamma}_s^{(k+1)},\bm{\gamma}_a^{(k+1)}\middle| \bm{\gamma}_s^{(k)},\bm{\gamma}_a^{(k)}\right)\!\ge\!\widetilde{R_s}\!\left(\!\bm{\gamma}_s^{(k)}\!,\!\bm{\gamma}_a^{(k)}\!\middle|\!\bm{\gamma}_s^{(k)},\!\bm{\gamma}_a^{(k)}\!\right)\!\ge\!\widetilde{R_s}\left(\bm{\gamma}_s^{(k)},\bm{\gamma}_a^{(k)}\middle| \bm{\gamma}_s^{(k-1)},\bm{\gamma}_a^{(k-1)}\right)$, which show that the optimal rates $\widetilde{R_s}$ obtained in the sequence of the sub-problems (\ref{eqR_sec_tilde}) is monotonically increasing. Since the achievable secrecy rate is upper-bounded for a finite transmission power, Algorithm~\ref{alg1} will converge to a fixed-point solution. This establishes the convergence of the proposed secrecy rate maximization algorithm. We will also confirm the rate of convergence of Algorithm~\ref{alg1} with numerical results in the next subsection.

\subsection{Numerical Results}

To illustrate the convergence of Algorithm~\ref{alg1}, we present the sequence of the optimized rates $\left\{\widetilde{R_s}\left(\bm{\gamma}_s^{(k)},\bm{\gamma}_a^{(k)}\middle|\bm{\gamma}_s^{(k-1)},\bm{\gamma}_a^{(k-1)}\right)\right\}_{k\ge 1}$ achieved in each step of the iterative power optimizations. We consider two cases of cooperative MIMO channels in presence of an eavesdropper with multiple possible locations. In the first case, we assume $K=N=2$ legitimate transmit and receive nodes, where the channel coefficients of the main channel $\mb{H}_1$ are given as follows:
\begin{equation}
\mb{H}_1 = \left[\begin{array}{cc}
1.97 - 0.92i & 0.98 + 0.47i\\
-0.63-0.035i & 0.019 - 1.24i
\end{array}\right].\label{eqH1}
\end{equation}

In the second case, we assume $K=N=3$ with the channel coefficients given as:
\begin{equation}
\mb{H}_2\! =\! \left[\!\begin{array}{ccc}
\!-1.06-1.65i\! &\! 3.01+0.11i\! &\! -0.08-0.60i\!\\
\!0.09+0.72i\! &\! -0.72-0.59i\! &\! -1.81+0.46i\!\\
\!0.53-0.66i\! &\! 0.17+0.28i\! &\! -0.35+0.59i\!
\end{array}\!\right]\!.\label{eqH2}
\end{equation}
In both cases, the eavesdropper is equipped with 2 receive antennas and can appear at either $L=10$ or at 20 possible locations, which are evenly distributed on a circle with equal distance $d_E=30$ meters towards the transmit head node. In Fig.~\ref{figConverge}, the secrecy rate obtained in each iteration of the power optimization is plotted for the cooperative MIMO channels $\mb{H}_1$ and $\mb{H}_2$. As a comparison, the optimal achievable secrecy rates are also obtained by exhaustive search over the space of the power allocation vectors. As shown in the Fig.~\ref{figConverge}, over the channels $\mb{H}_1$ and $\mb{H}_2$, the secrecy rate converges fast to the corresponding optimal values in about 4 iterations. In both cases, the secrecy rates monotonically increase, which is in line with our prediction in Section~\ref{secOptAlg}. In addition, we also observe that a larger number of Eve's locations does not change the resulting optimal secrecy rate. Therefore, we fix $L=10$ in the following numerical evaluations.

\begin{figure}[t!]
\centerline{\includegraphics[width=2.8in]{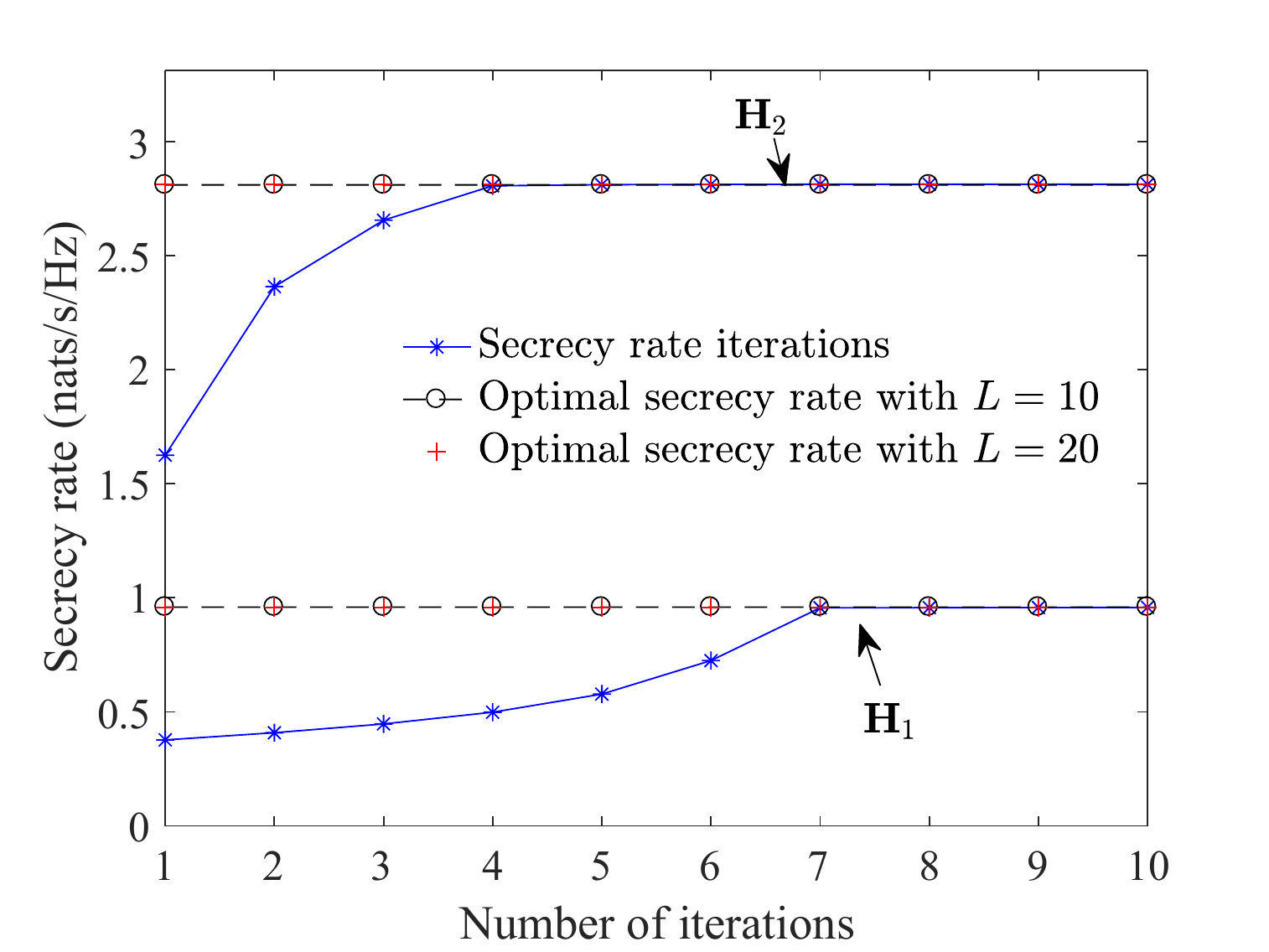}
}
\caption{Convergence of the secrecy rate for the iterative power optimization Algorithm~\ref{alg1}. The sample channels between the legitimate nodes are chosen as $\mb{H}_1$ in (\ref{eqH1}) and $\mb{H}_2$ in (\ref{eqH2}).}
\label{figConverge}
\end{figure}

In Fig.~\ref{figIter}, the required number of iterations of Algorithm~\ref{alg1} is plotted for the cooperative MIMO system with $K=N=2$ or $K=N=4$ transmit and receive nodes, distributed within their corresponding clusters with radius 5 meters. The number of antennas at Eve is set to $M=2$. The distance between head nodes $d_B$ and the distance $d_E$ are set to 30 meters. We evaluate Algorithm~\ref{alg1} for 50 random samples of the legitimate channel $\mb{H}$ and set the stopping criterion of the algorithm as $\epsilon=0.1$. When $K=N=2$, the maximum number of iteration of Algorithm~\ref{alg1} is 10, and for most cases, Algorithm~\ref{alg1} converges within 7 iterations. When $K=N=4$, only 2 or 3 iterations are needed for most of the cases. Note that in each iteration, a simple convex sub-problem is solved and therefore, Algorithm~\ref{alg1} has low computational complexity. These results demonstrate the practicality of our approach.

\begin{figure}[t!]
\centerline{\includegraphics[width=2.8in]{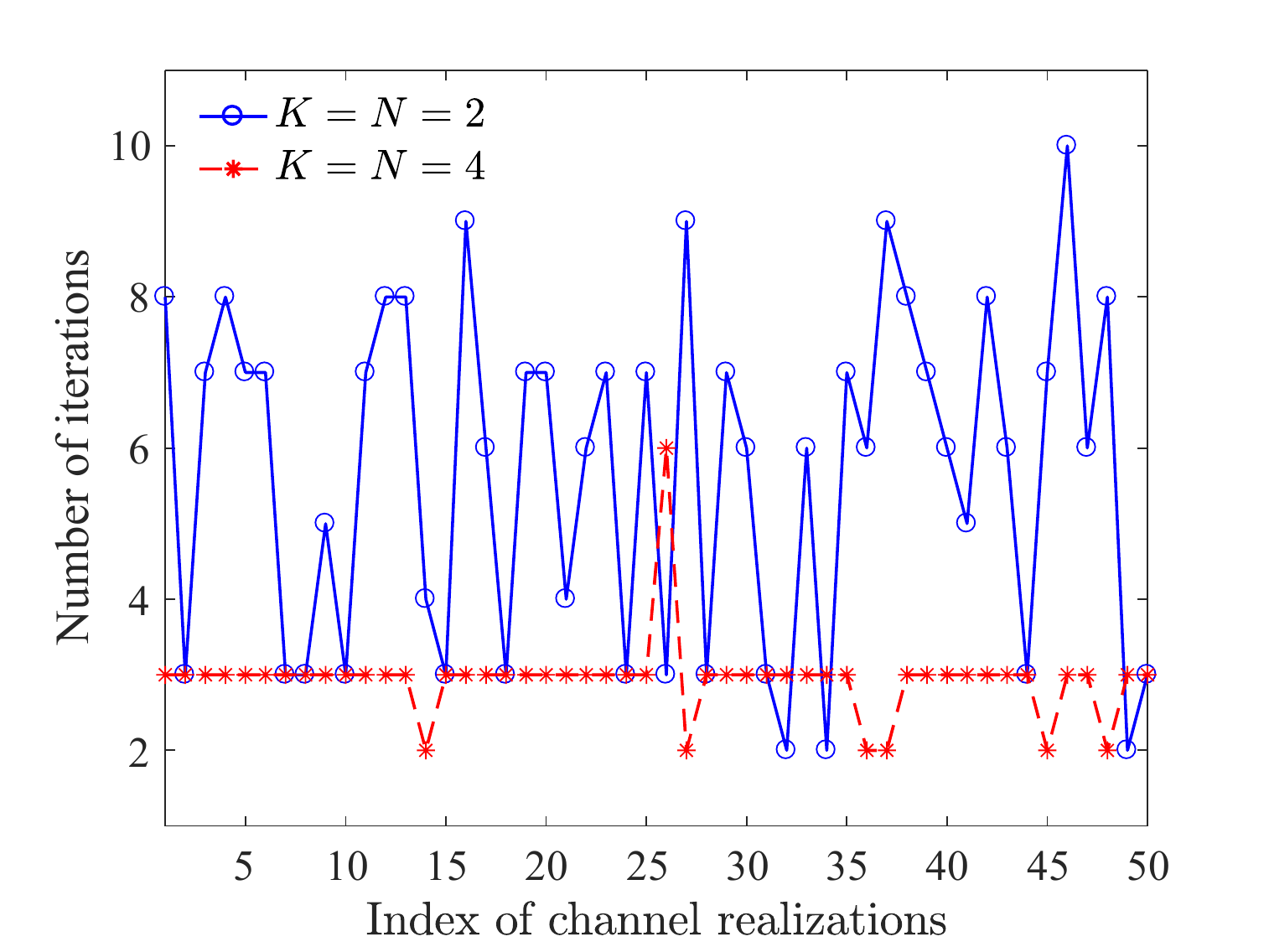}
}
\caption{Number of iterations of Algorithm~\ref{alg1} with stopping criterion $\epsilon=0.1$.}
\label{figIter}
\end{figure}

Next, we investigate the secrecy rate when the cluster heads can flexibly determine the number of cooperative nodes, which are located within a given cluster radius. We assume that there always exists sufficiently many assisting nodes to meet the head node request. The number of legitimate nodes in each cluster is allowed to vary from 1 to 5, while the number of antennas at Eve is fixed to 2. The distance between the legitimate head nodes is set to $d_B=30$ meters and the cooperative nodes are distributed in the corresponding clusters with radius 4 or 8 meters. We assume that the eavesdropper-free region is centered at the transmit head node with radius of either $d_E=20$ or 40 meters. As a comparison, we also plot the achievable secrecy rate obtained in~\cite{he2014}, where Eve could be located anywhere in the network. In this case, Fig.~\ref{figR_sec} shows that the achievable secrecy rate of~\cite{he2014} is zero for $K\le M$, and then increases along with $K$ for $K>M$. When there exists an eavesdropper-free region with a radius 20 meters, the secrecy rate is significantly improved for $K>M$, with 2.4 to 3.8 nats/s/Hz improvement as compared to the corresponding \lq\lq anywhere Eve\rq\rq\ case. When the radius of the eavesdropper-free region increases to 40 meters, the secrecy rate becomes linearly proportional to the number of transmitters for the whole range of $K$. In other words, by leveraging the location constraint, the outage can be eliminated, and the secrecy rate can be efficiently improved with more nodes joining in the cooperative clusters. Fig.~\ref{figR_sec} also shows that the secrecy rate can be also improved by increasing the cluster radius. Note that when $K=N=1$, only the transmit and the receive head nodes are activated, which becomes the conventional non-cooperative system. Results show that no reliable secrecy communication is possible for $d_E=20$, while only a marginal secrecy rate is achieved for $d_E=40$.

\begin{figure}[t!]
\centerline{\includegraphics[width=2.8in]{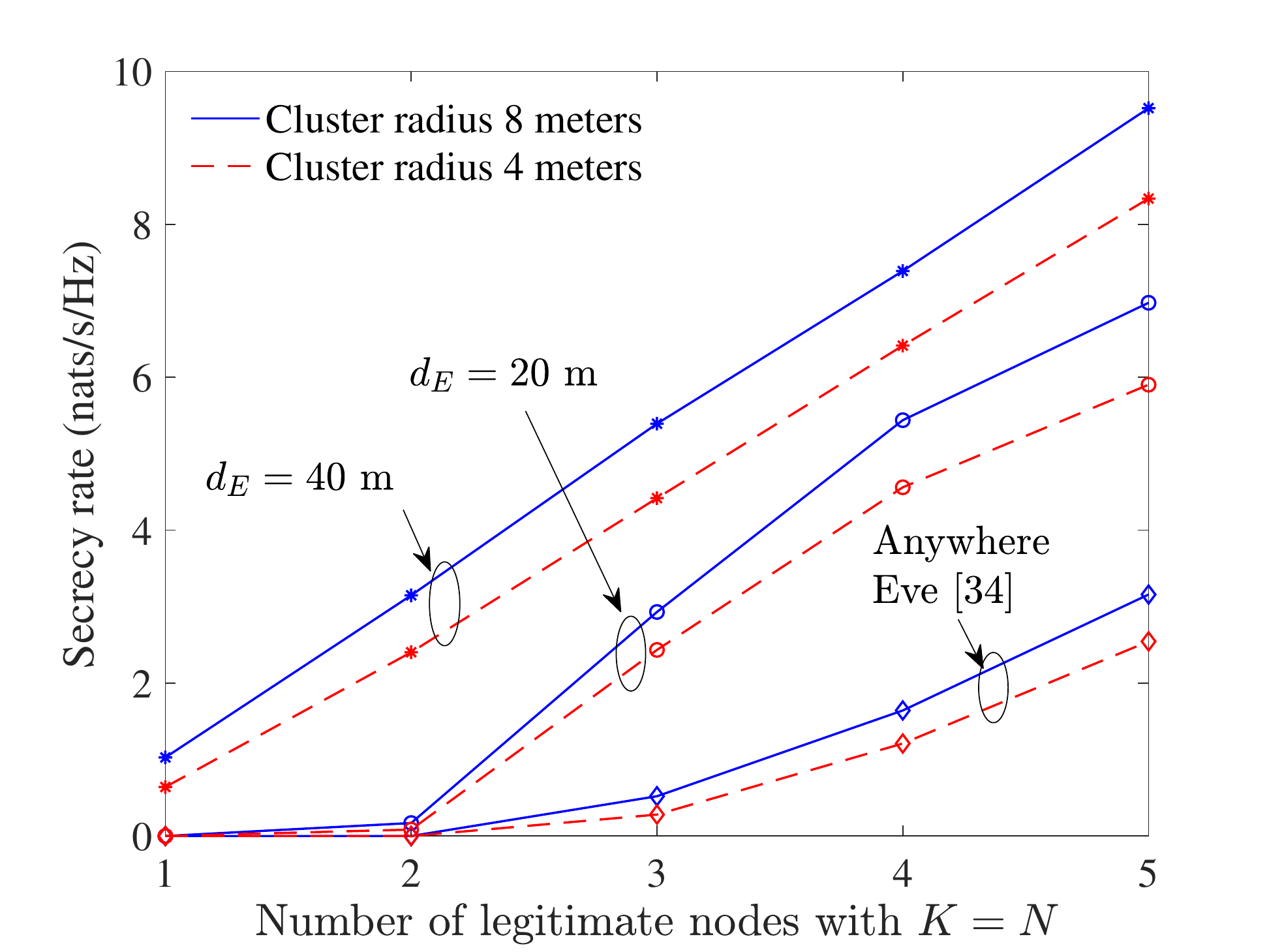}
}
\caption{Secrecy rate $R_\mathrm{sec}$ with equal number of legitimate nodes in transmit and receive clusters, where $M=2$, $d_B=30$ meters, and $d_E=20$ or 40 meters. Solid lines: secrecy rate with cluster radius 8 meters; dashed lines: secrecy rate with cluster radius 4 meters.}
\label{figR_sec}
\end{figure}

Next, we study in more details the impacts of the cluster radius on the secrecy rate of the cooperative MIMO. The secrecy rate is evaluated when $K$ legitimate transmit nodes and $N$ legitimate receive nodes are distributed in the corresponding cluster areas with the cluster radius ranging from 1 to 10 meters, where we set $K=4$ and $N=4$ or 6. The distance between the head nodes of Alice and Bob is set to $d_B=30$ meters. In Fig.~\ref{figR_sec_r}~(a), the number of antennas at Eve is $M=2$, less than both $K$ and $N$, where we also plot the non-zero secrecy rate achieved in~\cite{he2014}. In these settings, it is observed that a significant secrecy rate can be achieved even when $d_E<d_B$. This secrecy rate is achieved due to the degree-of-freedom advantage of the legitimate nodes as opposed to the eavesdropper. As the cluster radius increases, the legitimate nodes are distributed more dispersively within their corresponding clusters and, thus, can reduce the distance between some of the legitimate transmitters and receivers. Therefore, the secrecy rate improves as the cluster radius increases, which is also visible for $d_E=30$ meters. It is noted that although Bob with $N=6$ receive nodes achieves higher secrecy rate, as the radius increases, the rate improvement due to the radius of clusters is less significant compared to the case $N=4$. It is because the additional receive nodes are located further from Alice (e.g., the rightmost node of Bob in Fig.~\ref{figSystem}). Increasing the cluster radius also increases the distance between these nodes and the nodes in Alice, which results in smaller secrecy rate improvement. In Fig.~\ref{figR_sec_r}~(b), we compare the achievable secrecy rate when Eve has equal or more antennas compared to the number of legitimate nodes at Alice and Bob. Specifically, we assume $M=6$ antenna elements at Eve, while the distance between Alice and Eve is $d_E=30$ or 35 meters. In these settings, when the distances $d_B$ and $d_E$ are equal and $N=4$ (i.e., the legitimate receiver has inferior capability in terms of the number of antennas), only marginal secrecy rate can be obtained when the cluster radius is small. By increasing the cluster radius, the secrecy rate is improved, almost linearly proportional to the cluster radius. In addition, the secrecy rate can be improved more effectively for $d_E=35$ meters, i.e., when the legitimate receiver has distance advantage. As a comparison, by increasing both the cluster radius and the number of nodes at Bob, the secrecy rates can be improved by 1 nats/s/Hz for $d_E=30$ and by 2 nats/s/Hz for $d_E=35$.

\begin{figure}[t!]
\centering
\subfigure[]{\includegraphics[width=2.8in]{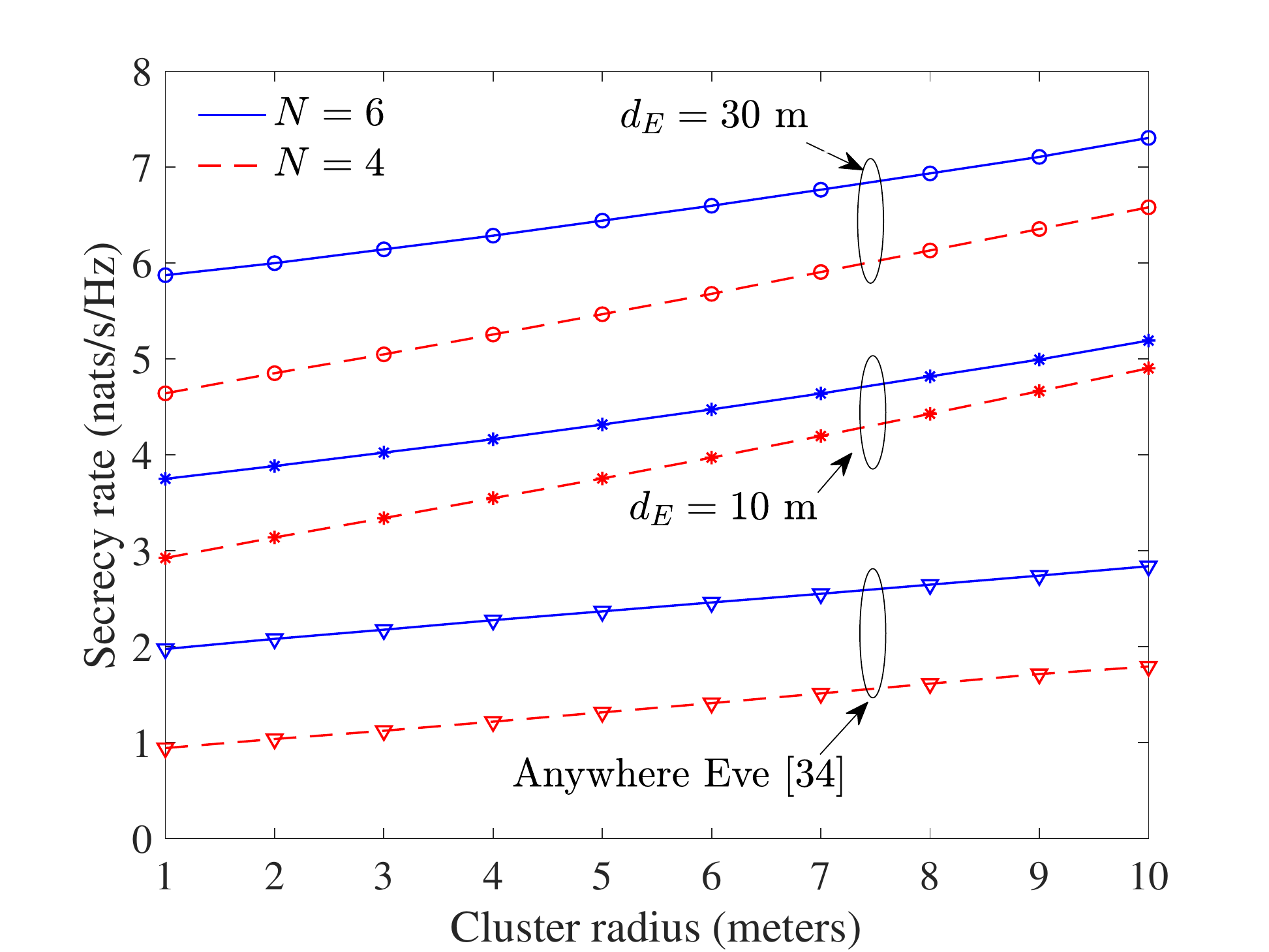}
}
\subfigure[]{\includegraphics[width=2.8in]{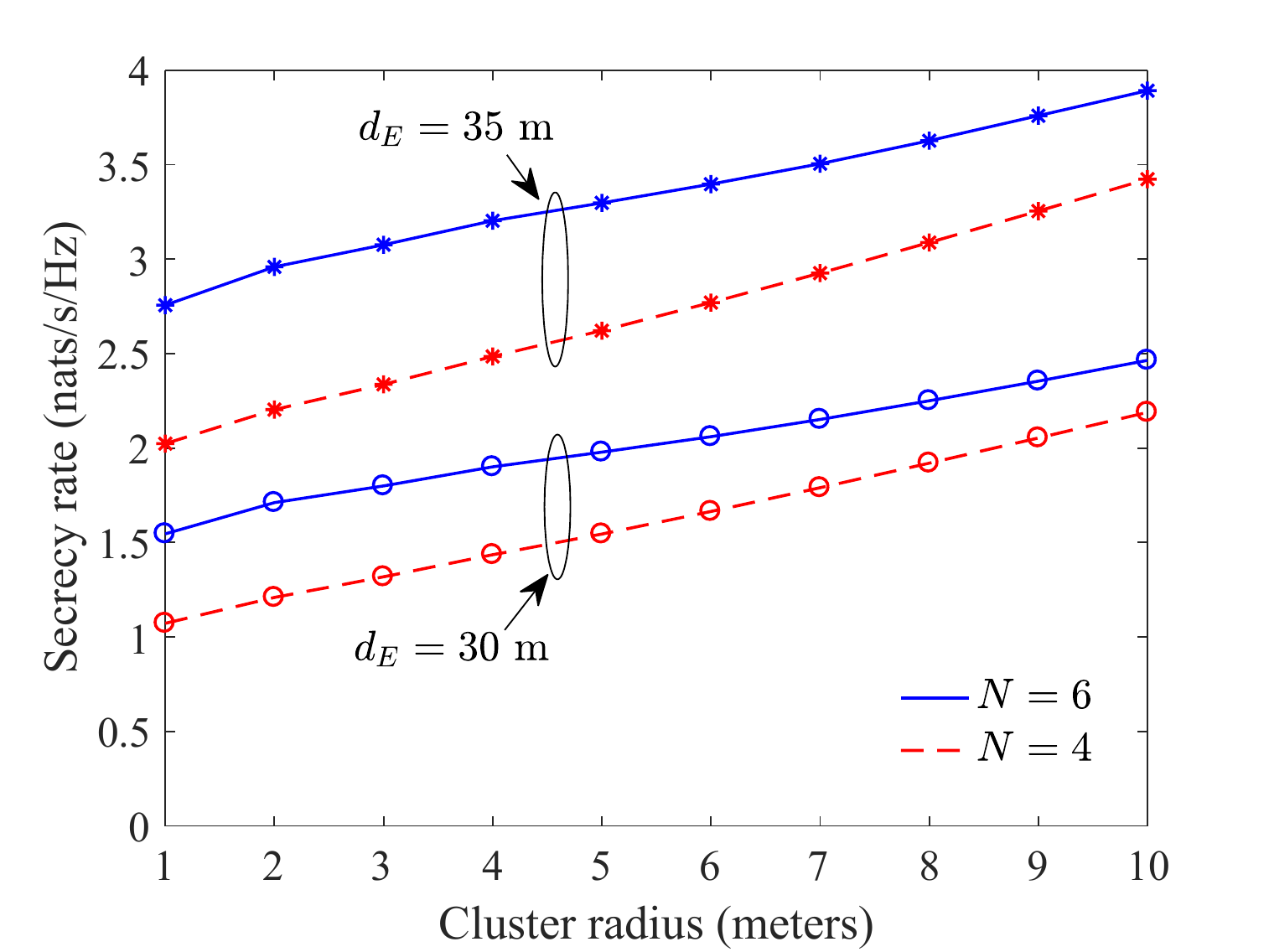}
}
\caption{Secrecy rate $R_\mathrm{sec}$ with $K=4$ and $N=4$ or 6, with distance between head nodes of Alice and Bob being $d_B=30$ meters. The number of antennas at Eve is $M=2$ in (a) and $M=6$ in (b). In (a), the distance between Alice and Eve $d_E=10$ or 30 meters, while in (b) $d_E=30$ or 35 meters.}
\label{figR_sec_r}
\end{figure}

\section{Conclusion}\label{secConclude}

As the small footprint mobile devices can be equipped only with limited number of antenna, the secrecy communications between such devices is difficult to realize when the eavesdropper has more antennas or experiences superior channel conditions. The proposed secrecy cooperative MIMO (i.e., Reconfigurable Distributed MIMO) architecture addresses this issue by temporally activating nearby trusted mobile devices to form cooperative cluster and jointly transmit or receive confidential message, where the communications between clusters resemble a distributed MIMO system. The secrecy cooperative MIMO architecture aims to enable and improve the secrecy transmissions, by activating a sufficiently large number of trusted devices, allowing to leverage the diverse channel conditions.

In this paper, the eigen-direction precoding is applied to construct the combined signal and artificial noise and the corresponding secrecy rate is obtained, when the eavesdropper may be located at arbitrary number of possible representative locations. By using Random Matrix Theory, we obtain accurate approximation for the average rate between Alice and Eve. The proposed approximations enable the affine representations of the information rate between Alice and Eve. By using the affine representations, the non-convex secrecy rate optimization for multiple possible locations of Eve is solved by successive convex approximations, which can be efficiently computed by standard convex optimization tools. Results show that the achievable secrecy rate can be improved quite significantly by leveraging the location constraint of Eve. The secrecy rate can be further improved by enabling additional trusted devices, especially when the devices are distributed more dispersively in each cooperative cluster, i.e., when Alice and/or Bob outperform Eve in terms of the available spatial degrees of freedom and the average transmission distance.


%

\appendices
\section{Proof of Proposition~\ref{prop1}}\label{appx1}

The proofs of Propositions~\ref{prop1} and~\ref{prop2} relies on the following lemmas.
\begin{lemma}(l'H\^opital's rule~\cite{simon2006}).\label{lemma1}
Consider the ratio of determinants of the form $\det[\{f_i(w_j)\}_{1\le i\le a,1\le j\le a}]/\Delta_a(\mb{w})$, where $\mb{w}=[w_1,\ldots,w_a]$. As $w_{b+1},\ldots,w_a$ ($b\le a$) approach zero, the limit of the ratio is given by
\begin{align}
&\lim_{w_{b+1},\ldots,w_a\rightarrow 0}\frac{\det[\{f_i(w_j)\}]}{\Delta_a(\mb{w})}\nonumber\\
& = \frac{\det\left[\left\{f_i^{(j-1)}(w)\middle|_{w=0}\right\}_{\substack{1\le i\le a\\1\le j\le a-b}} \left\{f_i(w_j)\right\}_{\substack{1\le i\le a\\1\le j\le b}}\right]}{\Delta_b(\mb{w})\prod_{i=1}^b w_i^{a-b}\prod_{j=1}^{a-b-1}j!},
\end{align}
where $f^{(m)}(w)$ denotes the $m^{th}$ derivative of $f(w)$.
\end{lemma}

\begin{lemma}(Generalized Andr\'{e}ief integral~\cite{kieburg2010}).\label{lemma2}
Consider the integral
\begin{align*}
\mc{J} = \int_{\mbb{C}^L} & \det\left[\begin{array}{c}\left\{r_{i,j}\right\}_{\substack{1\le i\le a,1\le j\le L+a}}\\
\left\{R_j(w_i)\right\}_{\substack{1\le i\le L,1\le j\le L+a}}\end{array}\right]\\
&\times \det\left[\begin{array}{c}\left\{s_{i,j}\right\}_{\substack{1\le i\le b,1\le j\le L+b}}\\
\left\{S_j(w_i)\right\}_{\substack{1\le i\le L,1\le j\le L+b}}\end{array}\right]\mathrm{d}w_1\ldots\mathrm{d}w_L,
\end{align*}
where the functions $R_j(\cdot)$ and $S_j(\cdot)$ are such that the integral is convergent. Then, the following identity holds:
\begin{align}
&\mc{J} = (-1)^{ab}L!\nonumber\\
&\det\left[
\arraycolsep=1pt\def\arraystretch{1}
	\begin{array}{cc} 
		\{0\}_{\substack{1\le i\le b\\1\le j\le a}} & \{s_{i,j}\}_{\substack{1\le i\le b\\1\le j\le L+b}} \\
		\{r_{j,i}\}_{\substack{1\le i\le L+a\\1\le j\le a}} & \left\{\int_{\mbb{C}}R_i(w)S_j(w)\mathrm{d}w\right\}_{\substack{1\le i\le L+a\\1\le j\le L+b}}
	\end{array}\right].
\end{align}

\end{lemma}

We are now ready to prove Proposition~\ref{prop1}. Consider the quantity
\begin{equation}
\phi(s) = \mbb{E}\left[\det\left(\mb{I} + \mb{\Sigma}^{1/2}\mb{W}^\dagger\mb{W}\mb{\Sigma}^{1/2}\mb{U}\mb{\Gamma}\mb{U}^\dagger\right)^s\right].\label{eqphi_a}
\end{equation}
The matrix $\mb{W}$ is complex Gaussian distributed with the density
\begin{equation}
f(\mb{W}) = \pi^{-M K}\exp\left(-\mathrm{tr}\mb{W}\mb{W}^\dagger\right).\label{eqfW}
\end{equation}
Inserting (\ref{eqfW}) into (\ref{eqphi_a}) and applying the change-of-variables $\mb{X}=\mb{W}\mb{\Sigma}^{1/2}$, we obtain
\begin{align}
&\phi(s) = \nonumber\\
& \frac{\int_{\mc{U}_K}\!\int_{\mc{M}_{M,K}}\!\det\left(\mb{I}+\mb{X}^\dagger\mb{X}\mb{U}\mb{\Gamma}\mb{U}^\dagger\right)^s\! e^{-\mathrm{tr}\mb{X}^\dagger\mb{X}\mb{\Sigma}^{-1}}\!\mathrm{d}\mb{X}\!\mathrm{d}\mb{U}}{\det[\mb{\Sigma}]^M\int_{\mc{M}_{M,K}}\exp\left(-\mathrm{tr}\mb{X}^\dagger\mb{X}\right)\mathrm{d}\mb{X}},\label{eqphi_b}
\end{align}
where $\mc{M}_{M,K}$ denotes the space of $M\times K$ complex matrices, $\mathrm{d}\mb{U}$ denotes the normalized Haar measure on the unitary group $\mc{U}_{K}$, and $\mathrm{d}\mb{X}$ defines the measure $\mathrm{d}\mb{X}=\prod_{i=1}^M\prod_{j=1}^K \Re(\mb{X}_{i,j})\Im(\mb{X}_{i,j})$. The denominator of (\ref{eqphi_b}) normalizes the right-hand-side of (\ref{eqphi_b}). 

Next, we apply the eigenvalue decomposition $\mb{X}^\dagger\mb{X}=\mb{V}\mb{\Omega}\mb{V}^\dagger$ with the Jacobian given by~\cite{mathai1997} as $\mathrm{d}\mb{X}=\Delta_K(\bm{\omega})^2\prod_{i=1}^K \omega_i^{M-K} \mathrm{d}\bm{\omega}\mathrm{d}\mb{V}$, where $\mathrm{d}\bm{\omega}=\mathrm{d}\omega_1\ldots\mathrm{d}\omega_K$ and $\mathrm{d}\mb{V}$ is the normalized Haar measure. Then, $\phi(s)$ can be rewritten as
\begin{align}
\frac{\int_{[0,\infty)^K}\!\frac{\Delta_K(\bm{\omega})^2}{\prod_{i=1}^K\omega_i^{K-M}}\!\int_{\mc{U}_{K}}e^{-\mathrm{tr}\mb{V}\mb{\Omega}\mb{V}^\dagger\mb{\Sigma}^{-1}}\mc{I}_1(\mb{\Omega},\mb{V}) \mathrm{d}\mb{V}\!\mathrm{d}\bm{\omega}}{\det[\mb{\Sigma}]^M\int_{[0,\infty)^K}\Delta_K(\bm{\omega})^2\prod_{i=1}^K\omega_i^{M-K}\exp(-\omega_i)\mathrm{d}\bm{\omega}},\label{eqphi_c}
\end{align}
where
\begin{equation}
\mc{I}_1(\mb{\Omega},\mb{V}) = \int_{\mc{U}_K}\det\left(\mb{I}+\mb{V}\mb{\Omega}\mb{V}^\dagger\mb{U}\mb{\Gamma}\mb{U}^\dagger\right)^s\mathrm{d}\mb{U}.\label{eqI1_a}
\end{equation}
We first solve the integral $\mc{I}_1(\mb{\Omega},\mb{V})$ by assuming all the diagonal elements of $\mb{\Gamma}$ being non-zero, i.e., $n=K$. The general result with $0<n<K$ is obtained by taking the limit $\gamma_{n+1},\ldots,\gamma_K\rightarrow 0$. By using the integral identity~\cite[Eq. (3.21)]{gross1989}, $\mc{I}_1(\mb{\Omega},\mb{V})$ can be solved as
\begin{align}\label{eqI1_b}
\frac{\det[(1+\omega_i\gamma_j)^{s+K-1}]}{\Delta_K(\bm{\omega})\Delta_K(\bm{\gamma})}\prod_{j=0}^{K-1}\frac{\Gamma(s+K-j)\Gamma(j+1)}{\Gamma(s+K)},
\end{align}
where the row index $i$ and the column index $j$ in the determinant of (\ref{eqI1_b}) run from 1 to $K$. Using Lemma~\ref{lemma1}, $\mc{I}_1(\mb{\Omega},\mb{V})$ is obtained, with $M\ge K\ge n$, by letting $\gamma_{n+1},\ldots,\gamma_K\rightarrow 0$ as
\begin{align}
&\mc{I}_1(\mb{\Omega},\mb{V}) = \prod_{j=K-n}^{K-1}\frac{\Gamma(s+K-j)\Gamma(j+1)}{\Gamma(s+K)}\nonumber\\
&\frac{\det\left[\left\{\omega_i^{j-1}\right\}_{\substack{1\le i\le K\\ 1\le j\le K-n}} \left\{(1+\omega_i\gamma_j)^{s+K-1}\right\}_{\substack{1\le i\le K\\ 1\le j\le n}}\right]}{\Delta_K(\bm{\omega})\Delta_n(\bm{\gamma})\prod_{i=1}^n\gamma_i^{K-n}}.\label{eqI1_c}
\end{align}
Therefore, the integral $\mc{I}_1(\mb{\Omega},\mb{V})$ is independent of the matrix $\mb{V}$, i.e., $\mc{I}_1(\mb{\Omega})\equiv\mc{I}_1(\mb{\Omega},\mb{V})$, which can be pulled out of the integral over $\mb{V}\in\mc{U}_K$ in (\ref{eqphi_c}). The integration over $\mb{V}$, denoted as $\mc{I}_2(\mb{\Omega})=\int_{\mc{U}_K}\exp(-\mathrm{tr}\mb{V}\mb{\Omega}\mb{V}^\dagger\mb{\Sigma}^{-1})\mathrm{d}\mb{V}$, can be solved by the Harish-Chandra-Itzykson-Zuber integral formula~\cite{itzykson1980} as follows
\begin{align}
\mc{I}_2(\mb{\Omega}) = \frac{\det[\mb{\Sigma}]^{K-1}}{\Delta_K(\bm{\omega})\Delta_K(\bm{\sigma})}\det\left[e^{-\frac{\omega_i}{\sigma_j}}\right]\prod_{j=0}^{K-1}\Gamma(j+1).\label{eqHCIZ}
\end{align}
Moreover, the denominator in (\ref{eqphi_c}) is a Selberg integral~\cite[Eq. (17.6.5)]{mehta2004} and solved as
\begin{align}
\int\frac{\Delta_K(\bm{\omega})^2}{\prod_{i=1}^K\omega_i^{K-M}e^{\omega_i}}\mathrm{d}\bm{\omega} = \prod_{i=1}^K \Gamma(j+1)\Gamma(M-K+j).\label{eqSelberg}
\end{align}
When $K>M$, (\ref{eqSelberg}) also holds by exchanging $M$ and $K$. Inserting (\ref{eqI1_c})-(\ref{eqSelberg}) into (\ref{eqphi_c}), $\phi(s)$ can be derived as in (\ref{eqphi_d}) on top of the next page, where the second equality is due to Lemma~\ref{lemma2}. The integrals in the determinant of (\ref{eqphi_d}) can be represented as the generalized hypergeometric function as
\begin{align*}
\int_0^\infty & x^{M-K}(1+x\gamma_j)^{s+K-1}e^{-\frac{x}{\sigma_i}}\mathrm{d}x = \frac{\Gamma(M-K+1)}{\sigma_i^{K-M-1}} \nonumber\\
& \times {}_2F_0\left(
\begin{array}{c}
M-K+1,1-s-K\\ - 
\end{array}
\middle| -\sigma_i\gamma_j
\right).
\end{align*}
By setting $s=1$, and factoring out $\sigma_i^{M-K+1} \Gamma(M-K+j)$ from the $1^{st}$ to the $(K-n)^{th}$ columns, $\Gamma(M-K+1) \sigma_i^{M-K+1}$ from the last $n$ columns of the determinant on the right-hand-side of (\ref{eqphi_d}), we obtain the desired result as in (\ref{eqfe_tilde1}).

\begin{figure*}[!t]
	\normalsize
	\setcounter{equation}{40}
	
\begin{align}
\phi(s) &= \frac{1}{K!} \left(\prod_{j=K-n}^{K-1}\frac{\Gamma(s+K-j)\Gamma(j+1)}{\Gamma(s+K)}\right)\left(\prod_{j=1}^K\frac{1}{\Gamma(M-K+j)}\right)\frac{\det[\mb{\Sigma}]^{K-M-1}}{\Delta_n(\bm{\gamma})\Delta_K(\bm{\sigma})\prod_{i=1}^n\gamma_i^{K-n}}\nonumber\\
&\qquad\times \int_{[0,\infty)^K}\prod_{i=1}^K\omega_i^{M-K}\det\left[\begin{array}{cc}
\left\{\omega_i^{j-1}\right\}_{\substack{1\le i\le K\\1\le j\le K-n}} & \left\{(1+\omega_i\gamma_j)^{s+K-1}\right\}_{\substack{1\le i\le K\\1\le j\le n}}
\end{array}
\right]\det[e^{-\frac{\omega_i}{\sigma_j}}]\mathrm{d}\bm{\omega}\nonumber\\
&=\left(\prod_{j=K-n}^{K-1}\frac{\Gamma(s+K-j)\Gamma(j+1)}{\Gamma(s+K)}\right)\left(\prod_{j=1}^K\frac{1}{\Gamma(M-K+j)}\right)\frac{\det[\mb{\Sigma}]^{K-M-1}}{\Delta_n(\bm{\gamma})\Delta_K(\bm{\sigma})\prod_{i=1}^n\gamma_i^{K-n}}\nonumber\\
&\qquad\times \det\left[
\begin{array}{cc}
\left\{\sigma_i^{M-K+j}\Gamma(M-K+j)\right\}_{\substack{1\le i\le K\\1\le j\le K-n}} & \left\{\int_0^\infty x^{M-K}(1+x\gamma_j)^{s+K-1}e^{-\frac{x}{\sigma_i}}\mathrm{d}x\right\}_{\substack{1\le i\le K\\1\le j\le n}}
\end{array}
\right],\label{eqphi_d}
\end{align}
	
	\setcounter{equation}{\value{MYtempeqncnt}}
	\hrulefill
	\vspace*{4pt}
\end{figure*}
\setcounter{equation}{41}

\section{Proof of Proposition~\ref{prop2}}\label{appx2}

By following the same procedures as in (\ref{eqphi_a})-(\ref{eqphi_c}), we obtain $\phi(s)$ as
\begin{align}
\frac{\int_{[0,\infty)^M}\!\frac{\Delta_K(\bm{\omega})^2}{\prod_{i=1}^M\omega_i^{M-K}}\!\int_{\mc{U}_{K}}\!e^{-\mathrm{tr}\mb{V}\mb{\Omega}\mb{V}^\dagger\mb{\Sigma}^{-1}}\mc{I}_1(\mb{\Omega},\mb{V}) \mathrm{d}\mb{V}\!\mathrm{d}\bm{\omega}}{\det[\mb{\Sigma}]^M\int_{[0,\infty)^M}\Delta_K(\bm{\omega})^2\prod_{i=1}^M\omega_i^{K-M}\exp(-\omega_i)\mathrm{d}\bm{\omega}},\label{eqphi_e}
\end{align}
where $\mc{I}_1(\mb{\Omega},\mb{V})$ is given in (\ref{eqI1_a}). When $K>M$, the diagonal matrix $\mb{\Omega}$, containing the $M$ non-zero eigenvalues of the Hermitian matrix $\mb{\Sigma}^{1/2}\mb{W}^\dagger\mb{W}\mb{\Sigma}^{1/2}$, is of the form $\mb{\Omega}=\mathrm{diag}([\omega_1,\ldots,\omega_M,0,\ldots,0])$. Before applying Lemma~\ref{lemma1} to (\ref{eqI1_c}) to set the corresponding $\omega_{M+1},\ldots,\omega_K$ to zero, we have to rewrite $(1+\omega_i \gamma_j)^{s+K-1}$ in (\ref{eqI1_c}) to guarantee the convergence of the integral (\ref{eqphi_e}). By the generalized binomial expansion, we have
\begin{equation}\label{eqBinomial}
(1+\omega_i\gamma_j)^{s+K-1} = \sum_{l = 0}^\infty \frac{\Gamma(s+K)(\omega_i\gamma_j)^l}{\Gamma(s+K-l)\Gamma(l+1)},
\end{equation}
where we assume $|\omega_i \gamma_j|<1$ to guarantee the convergence. We will later on extend this expression to arbitrary values of $\omega_i \gamma_j$. Inserting (\ref{eqBinomial}) into (\ref{eqI1_c}), we obtain
\begin{align}
&\mc{I}_1(\mb{\Omega},\mb{V}) = \frac{\prod_{j=K-n}^{K-1}\Gamma(s+K-j)\Gamma(j+1)}{\Delta_K(\bm{\omega})\Delta_n(\bm{\gamma})}\nonumber\\
&\quad\times\det\left[
\begin{array}{c}
\left\{\omega_j^{i-1}\right\}_{\substack{1\le i\le K-n\\ 1\le j\le K}} \\ \left\{\sum_{l = 0}^\infty \frac{\omega_j^{K-n}(\omega_j\gamma_i)^l}{\Gamma(s+n-l)\Gamma(K-n+l+1)}\right\}_{\substack{1\le i\le n\\ 1\le j\le K}}
\end{array}
\right],\label{eqI1_d}
\end{align}
where the first $K-n$ summands in the infinite summations are cancelled as they are the linear combinations of $\{\omega_j^{i-1}\}_{1\le i\le K-n,1\le j\le K}$. Then, we apply Lemma~\ref{lemma1} to (\ref{eqI1_d}) to take the limit $\omega_{M+1},\ldots,\omega_K\rightarrow 0$. When $K>M\ge n$, we obtain
\begin{align*}
&\mc{I}_1(\mb{\Omega},\mb{V})\!=\!\prod_{j=K-n}^{K-1}\!\frac{\Gamma(s+K-j)\Gamma(j+1)}{\Gamma(s+M)\Gamma(K-M+1)}\frac{\prod_{j=1}^n\gamma_j^{n-M}}{\Delta_M(\bm{\omega})\Delta_n(\bm{\gamma})}\\
&\qquad\!\times\!\det\!\left[\!
\begin{array}{c}
\left\{\omega_j^{i-1}\right\}_{\substack{1\le i\le M-n\\ 1\le j\le M}} \\ \left\{{}_2F_1\left(
\begin{array}{c}
1,1-s-M\\K-M+1
\end{array}
\middle|-\omega_j\gamma_i
\right)\right\}_{\substack{1\le i\le n\\ 1\le j\le M}}
\end{array}
\right],
\end{align*}
and when $K>n>M$, we obtain
\begin{align*}
&\mc{I}_1(\mb{\Omega},\mb{V})\!=\!\frac{\prod_{j=n-M}^{n-1}\Gamma(s+n-j)\Gamma(j+K-n+1)}{\Delta_M(\bm{\omega})\Delta_n(\bm{\gamma})\Gamma(s+M)^M\Gamma(K-M+1)^M}\\
&\!\times\!\det\!\left[\!
\begin{array}{c}
\left\{\gamma_j^{i-1}\right\}_{\substack{1\le i\le n-M\\ 1\le j\le n}} \\ \left\{\gamma_j^{n-M}{}_2F_1\left(
\begin{array}{c}
1,1-s-M\\K-M+1
\end{array}
\middle|-\omega_i\gamma_j
\right)\right\}_{\substack{1\le i\le M\\ 1\le j\le n}}
\end{array}
\right],
\end{align*}
where we replace the infinite summation with its corresponding hypergeometric representation:
\begin{align*}
\sum_{l=0}^\infty & \frac{x^l}{\Gamma(a-l)\Gamma(b+l+1)}\nonumber\\
& = \frac{1}{\Gamma(a)\Gamma(b+1)}{}_2F_1\left(
\begin{array}{c}
1,1-a\\b+1
\end{array}
\middle| -x
\right).
\end{align*}
Note that the representation using the hypergeometric function can be analytically continued to arbitrary values of $x$. Again, we notice that the integral $\mc{I}_1(\mb{\Omega},\mb{V}) \equiv \mc{I}_1(\mb{\Omega})$ is independent of the matrix $\mb{V}$ and can be pulled out of the integral over $\mb{V}\in\mc{U}_K$ in (\ref{eqphi_e}).

When $K>M$, the integral $\mc{I}_2(\mb{\Omega})=\int_{\mc{U}_K}\exp\left(-\mathrm{tr}\mb{V}\mb{\Omega}\mb{V}^\dagger\mb{\Sigma}^{-1}\right)\mathrm{d}\mb{V}$ is obtained by applying Lemma~\ref{lemma1} to (\ref{eqHCIZ}) when taking the limit $\omega_{M+1},\ldots,\omega_K\rightarrow 0$. That is,
\begin{align*}
&\mc{I}_2(\mb{\Omega})\!= \frac{\prod_{j=K-M}^{K-1}\Gamma(j+1)\det[\mb{\Sigma}]^{K-1}}{\Delta_M(\bm{\omega})\Delta_K(\bm{\sigma})\prod_{j=1}^M\omega_i^{K-M}}\\
&\!\times\!\det\!\left[\!
\begin{array}{cc}
\left\{(-\sigma_i)^{1-j}\right\}_{\substack{1\le i\le K\\ 1\le j\le K-M}} & \left\{\exp\left(-\frac{\omega_j}{\sigma_i}\right)\right\}_{\substack{1\le i\le K\\ 1\le j\le M}}
\end{array}
\right].
\end{align*}

The denominator of (\ref{eqphi_e}) is solved by (\ref{eqSelberg}) by interchanging $K$ and $M$. Then, after inserting $\mc{I}_1(\mb{\Omega})$ and $\mc{I}_2(\mb{\Omega})$ into (\ref{eqphi_e}), we obtain $\phi(s)$ for $K>M\ge n$ as in (\ref{eqphi_f}) on top of the next page, where the second equality is obtained by applying Lemma~\ref{lemma2} and the identity~\cite[Eq. (7.811.1)]{gradshteyn2014}. When $K>n>M$, we obtain $\phi(s)$ as in (\ref{eqphi_g}), where we applied Lemma~\ref{lemma2} in the second equality.

\begin{figure*}[!t]
	\normalsize
	\setcounter{equation}{44}
	
\begin{align}
\phi(s) &= \frac{1}{M!} \left(\prod_{j=K-n}^{K-1}\frac{\Gamma(s+K-j)\Gamma(j+1)}{\Gamma(s+M)\Gamma(K-M+1)}\right)\frac{\det[\mb{\Sigma}]^{K-M-1}\prod_{j=1}^n\gamma_j^{n-M}}{\Delta_n(\bm{\gamma})\Delta_K(\bm{\sigma})\prod_{i=1}^{M-1}\Gamma(j+1)}\nonumber\\
&\qquad\times \int_{[0,\infty)^M}\det\left[
\begin{array}{c}
\left\{\omega_j^{i-1}\right\}_{\substack{1\le i\le M-n\\ 1\le j\le M}} \\ \left\{{}_2F_1\left(
\begin{array}{c}
1,1-s-M\\K-M+1
\end{array}
\middle|-\omega_j\gamma_i
\right)\right\}_{\substack{1\le i\le n\\ 1\le j\le M}}
\end{array}
\right]
\det\left[
\begin{array}{c}
\left\{(-\sigma_j)^{1-i}\right\}_{\substack{1\le i\le K-M\\ 1\le j\le K}} \\ \left\{\exp\left(-\frac{\omega_i}{\sigma_j}\right)\right\}_{\substack{1\le i\le M\\ 1\le j\le K}}
\end{array}
\right]\mathrm{d}\bm{\omega}\nonumber\\
&=\frac{\prod_{j=1}^n\gamma_j^{n-M}}{\Delta_n(\bm{\gamma})\Delta_K(\bm{\sigma})}\prod_{j=K-n}^{K-1}\frac{\Gamma(s+K-j)\Gamma(j+1)}{\Gamma(j-K+M+1)\Gamma(s+M)\Gamma(K-M+1)}\nonumber\\
&\qquad\times \det\left[
\begin{array}{cc}
\left\{\sigma_i^{j-1}\right\}_{\substack{1\le i\le K\\1\le j\le K-n}} & \left\{\sigma_i^{K-M}{}_3F_1\left(\begin{array}{c} 1,1,1-s-M\\ K-M+1 \end{array} \middle|-\sigma_i\gamma_j\right)\right\}_{\substack{1\le i\le K\\1\le j\le n}}
\end{array}
\right]\label{eqphi_f}
\end{align}
	
\hrulefill

\begin{align}
\phi(s) &= \frac{1}{M!} \left(\prod_{j=n-M}^{n-1}\frac{\Gamma(s+n-j)\Gamma(j+K-n+1)}{\Gamma(s+M)\Gamma(K-M+1)}\right)\frac{\det[\mb{\Sigma}]^{K-M-1}}{\Delta_n(\bm{\gamma})\Delta_K(\bm{\sigma})\prod_{j=0}^{M-1}\Gamma(j+1)}\nonumber\\
&\qquad\times \int_{[0,\infty)^M}\det\left[
\begin{array}{c}
\left\{\gamma_j^{i-1}\right\}_{\substack{1\le i\le n-M\\ 1\le j\le n}} \\ \left\{\gamma_j^{n-M}{}_2F_1\left(
\begin{array}{c}
1,1-s-M\\K-M+1
\end{array}
\middle|-\omega_i\gamma_j
\right)\right\}_{\substack{1\le i\le M\\ 1\le j\le n}}
\end{array}
\right]
\det\left[
\begin{array}{c}
\left\{(-\sigma_j)^{1-i}\right\}_{\substack{1\le i\le K-M\\ 1\le j\le K}} \\ \left\{\exp\left(-\frac{\omega_i}{\sigma_j}\right)\right\}_{\substack{1\le i\le M\\ 1\le j\le K}}
\end{array}
\right]\mathrm{d}\bm{\omega}\nonumber\\
&=\frac{(-1)^{(K-M)(n-M)}}{\Delta_n(\bm{\gamma})\Delta_K(\bm{\sigma})}\prod_{j=n-M}^{n-1}\frac{\Gamma(s+n-j)\Gamma(j+K-n+1)}{\Gamma(j+M-n+1)\Gamma(s+M)\Gamma(K-M+1)}\nonumber\\
&\qquad\times \det\left[
\begin{array}{cc}
\left\{0\right\}_{\substack{1\le i\le K-M\\1\le j\le n-M}} & \left\{\sigma_j^{i-1}\right\}_{\substack{1\le i\le K-M\\1\le j\le K}} \\
\left\{\gamma_i^{j-1}\right\}_{\substack{1\le i\le n\\1\le j\le n-M}} & \left\{\gamma_i^{n-M}\sigma_j^{K-M}{}_3F_1\left(\begin{array}{c} 1,1,1-s-M\\ K-M+1 \end{array} \middle|-\sigma_j\gamma_i\right)\right\}_{\substack{1\le i\le n\\1\le j\le K}}
\end{array}
\right]\label{eqphi_g}
\end{align}

	\setcounter{equation}{\value{MYtempeqncnt}}
	\hrulefill
	\vspace*{4pt}
\end{figure*}
\setcounter{equation}{46}


\ifCLASSOPTIONcaptionsoff
  \newpage
\fi



%
%


%

\begin{IEEEbiography}[{\includegraphics[width=1in,height=1.25in,clip,keepaspectratio]{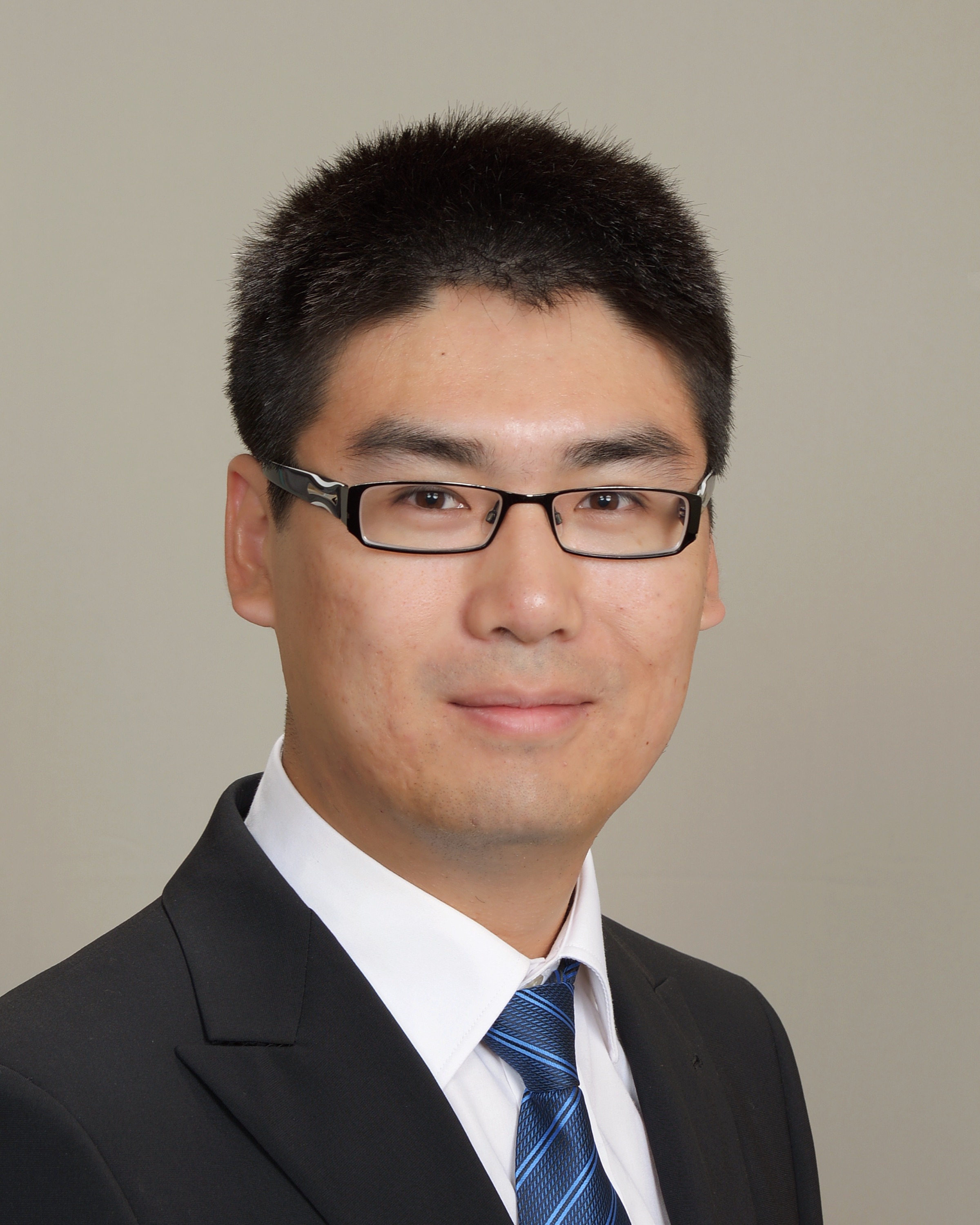}}]{Zhong Zheng (S'10--M'15)} received the B.Eng. degree from Beijing University of Technology, China in 2007, M.Sc. degree from Helsinki University of Technology, Finland in 2010, and D.Sc. degree from Aalto University, Finland in 2015. From 2015 to 2018, he held visiting positions at University of Texas at Dallas and National Institute of Standards and Technology. His research interests include massive MIMO, secure communications, millimeter wave communications, random matrix theory, and free probability theory.
\end{IEEEbiography}

\begin{IEEEbiography}[{\includegraphics[width=1in,height=1.25in,clip,keepaspectratio]{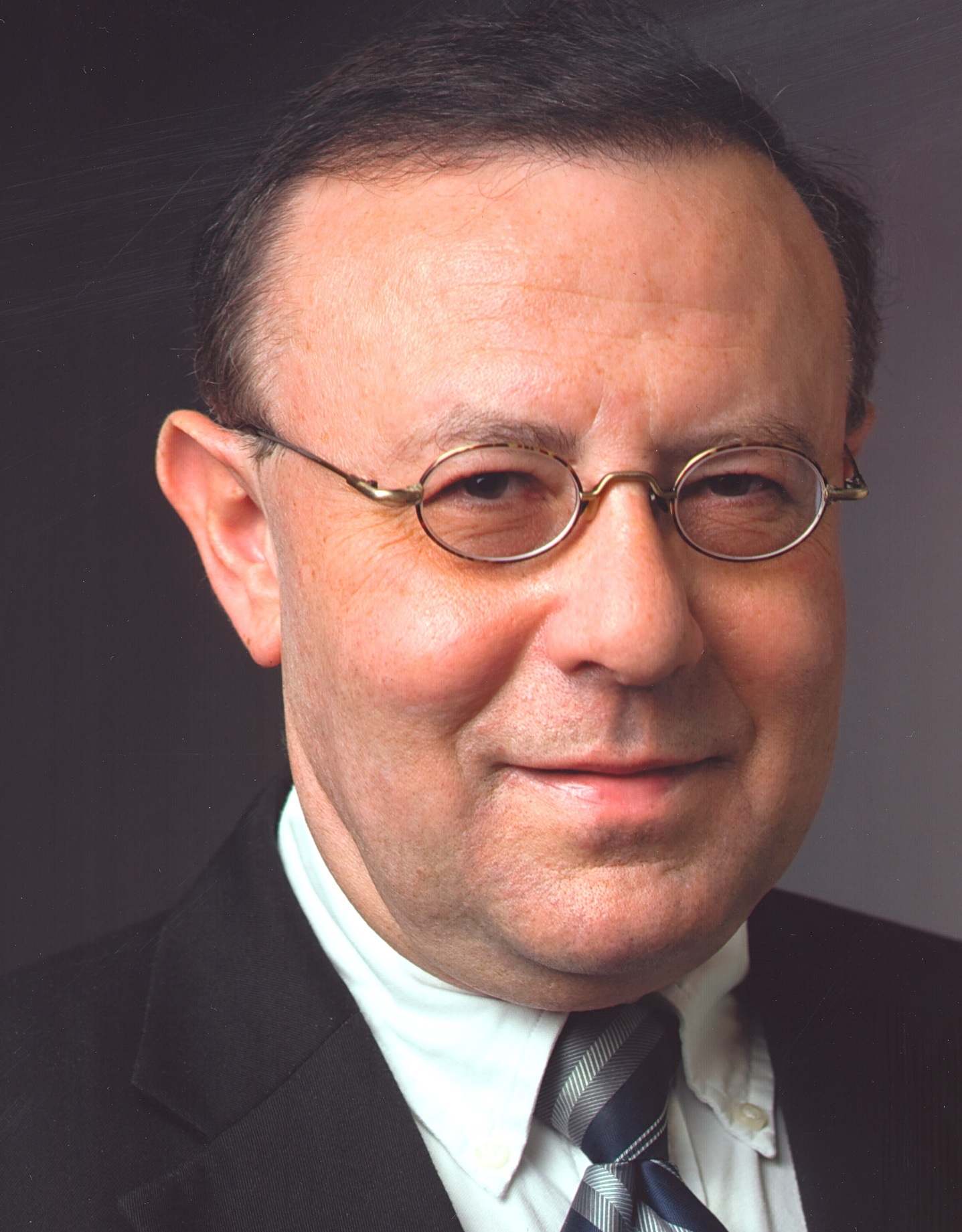}}]{Zygmunt J. Haas (S'84--M'88--SM'90--F'07)}
received the Ph.D. degree in electrical engineering from Stanford University, Stanford, CA, USA, in 1988. In 1988, he joined AT\&T Bell Laboratories in the Network Research Area, where he pursued research in wireless communications, mobility management, fast protocols, optical networks, and optical switching. In August 1995, he joined the Faculty of the School of Electrical and Computer Engineering, Cornell University, Ithaca, NY, USA, where he is a Professor. He heads the Wireless Network Laboratory, a research group with extensive contributions and international recognition in the area of ad hoc networks and sensor networks. He has authored over 200 technical conference and journal papers and holds 18 patents in the areas of wireless networks and wireless communications, optical switching and optical networks, and high-speed networking protocols. His research interests include protocols for mobile and wireless communication and networks, secure communications, and modeling and performance evaluation of large and complex systems. He has organized several workshops, delivered numerous tutorials at major IEEE and ACM conferences, and served as an Editor for several journals and magazines, including the IEEE/ACM TRANSACTIONS ON NETWORKING, the IEEE TRANSACTIONS ON WIRELESS COMMUNICATIONS, the IEEE Communications Magazine, and Wireless Networks (Springer). He has been a Guest Editor of the IEEE JOURNAL ON SELECTED AREAS IN COMMUNICATIONS issues and served as the Chair of the IEEE Technical Committee on Personal Communications. He was the recipient of a number of awards and distinctions, including best paper awards, the 2012 IEEE ComSoc WTC Recognition Award for \lq\lq outstanding achievements and contribution in the area of wireless communications systems and networks,\rq\rq\ and the 2016 IEEE ComSoc AHSN Recognition Award for \lq\lq outstanding contributions to securing ad hoc and sensor networks.\rq\rq
\end{IEEEbiography}

\begin{IEEEbiography}[{\includegraphics[width=1in,height=1.25in,clip,keepaspectratio]{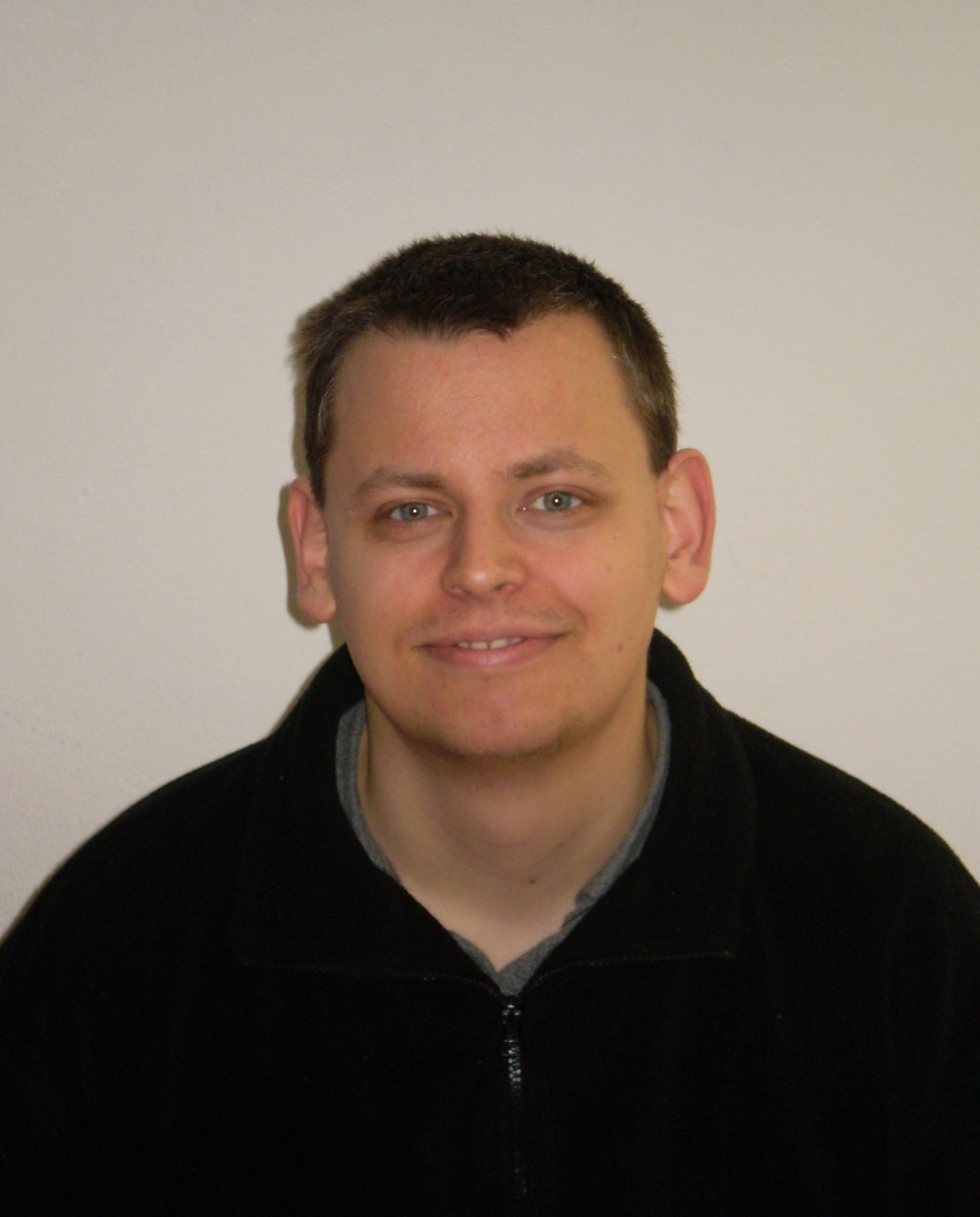}}]{Mario Kieburg}
received his Ph.D. in physics at the University of Duisburg-Essen, Germany, in 2010. After a two-year scholarship (2011-2013) at the State University of New York at Stony Brook, USA, he went to Bielefeld University, Germany, working and teaching as Senior Post-Doctoral Researcher in mathematical physics. In 2015, he habilitated in theoretical physics. His Research interest are the theoretical development and application of random matrix theory, in particular in the topics of harmonic analysis, group and representation theory, supersymmetry, orthogonal polynomial theory, quantum field theory, chaotic and disordered systems, time series analysis, and telecommunications.
\end{IEEEbiography}




\end{document}